\documentclass[11pt]{llncs}
\usepackage{amssymb,,graphicx,color}
\usepackage{amsmath}
\usepackage{amsfonts}
\usepackage{epsfig,latexsym,graphicx}
\usepackage{nameref}
\usepackage{algorithm2e}

 \newtheorem{prop}{Proposition}
 \newtheorem{lem}{Lemma}

\newcommand{\mt}[1]{\mathtt{#1}}
\newcommand{\ds}{\, \mbox{d}\mt{s}}

\newcommand{\cond}{|\!|}
\newcommand{\femp}{f^{\textit{emp}}_T}
\newcommand{\mfemp}{\mathop{f^{\textit{emp}}_T}}
\newcommand{\mbt}[1]{\boldsymbol{\mathtt{#1}}}
\newcommand{\domg}{\, \mbox{d}\mbt{\omega}}

\numberwithin{equation}{section}
\numberwithin{theorem}{section}
\numberwithin{prop}{section}
\numberwithin{lem}{section}

\begin{document}

\title{A spatio-spectral hybridization for edge preservation and noisy image restoration via local parametric mixtures and Lagrangian relaxation}
\author{Kinjal Basu$^1$ \and Debapriya Sengupta$^2$}
\institute{Department of Statistics,\\
Stanford University, USA.\\
Email : kinjal@stanford.edu
\vspace{2mm}
\and
Applied Statistics Unit, \\
Indian Statistical Institute, India.\\
Email : dps@isical.ac.in
}
\maketitle

\begin{abstract}\label{abstract}
This paper investigates a fully unsupervised statistical method for edge preserving image restoration and compression using a spatial decomposition scheme. Smoothed maximum likelihood is used for local estimation of edge pixels from mixture parametric models of local templates. For the complementary smooth part the traditional $L_2$-variational problem is solved in the Fourier domain with Thin Plate Spline (TPS) regularization  \cite{Grace}. It is well known that naive Fourier compression of the whole image fails to restore a piece-wise smooth noisy image satisfactorily due to Gibbs phenomenon \cite{Hewitt}. Images are interpreted as relative frequency histograms of samples from  bi-variate densities where the sample sizes might be unknown. The set of discontinuities is assumed to be completely unsupervised Lebesgue-null, compact subset of the plane in the continuous formulation of the problem. Proposed spatial decomposition uses a widely used topological concept, {\it partition of unity \/} \cite{Munkres,Tu}. The decision on edge pixel neighborhoods are made based on the multiple testing procedure of \cite{Holms,TSH}. Statistical summary of the final output is decomposed into two layers of information extraction, one for the subset of edge pixels and the other for the smooth region. Robustness is also demonstrated by applying the technique on noisy degradation of clean images. 
\end{abstract}

\section{Introduction}\label{intro}

\subsection{Data generation process} 
A monochromatic digital image can be regarded as a surface of the image intensity function defined at pixels of the image, with possible edges visually demarcating outlines of embedded objects in an otherwise smooth background. In a variety of applications of image analysis it is important to recognize presence of objects of interests in images. Identification of edge pixels plays an important role towards this objective. Edges can be classified into different categories depending on the behavior of the image in its neighborhood. Because of this reason edges are local features of an image (see \cite{Pratt}). Edge detection and image restoration problems in image processing are closely related to the jump surface estimation problem in statistics \cite{Qiu,Qiu2}.

In this article we consider an image as a bi-variate histogram data with pixels as bins where samples are drawn from a bi-variate density $f$ defined on $ \mathcal{D}=[0, 1] \times [0, 1]$. The unknown density $f$ is assumed to lack global smoothness properties and may have various forms of discontinuities like jumps and edges. Let $ \mathtt{s} = (x, y)\in \mathcal{D}$ denote a generic point in $\mathcal{D}$ and suppose we have a sample $\xi_1, \xi_2, \ldots, \xi_T$ from the unknown density $f$. Without loss of generality consider an $M\times M$ equispaced discretization of the domain $\mathcal{D}$ into $D=M^2$ pixels (squares) centered at \textit{sites} $\mt{s}_1, \mt{s}_2, \ldots, \mt{s}_D$ where $\mt{s}_i$'s are some suitable enumeration of the (\textit{sites}) given by

$$ \mathcal{G} = \left\{ \left(\frac{i+1/2}{M},\frac{j+1/2}{M} \right): 0 \leq i,j \leq (M-1)   \right\}.
$$
We shall denote the pixels by $\mt{S}_i \subset \mathcal{D}$. The choice of boundary of pixels $S_i$ is not crucial to our analysis as we work with continuous probability distributions. The collection of vertical and horizontal lines of the basic grid form a set of measure $0$. However, while implementing the algorithm it is convenient to work with a consistent convention. For theoretical discussions we shall assume that $\mt{S}_i$'s are open rectangles of the form $ (i, i+1) \times (j, j+1)$ and hence they are disjoint. The set $C = \mathcal{D}\backslash \left(\cup_{i \in \mathcal{G}}\, \mt{S}_i \right)$, consisting of the vertical and horizontal grid boundaries, is a set of measure $0$ . Finally, let $Y(\mt{s}_1), Y(\mt{s}_2), \ldots, Y(\mt{s}_D)$ denote the counts in different subsets $\mt{S}_1, \mt{S}_2, \ldots, \mt{S}_D$ respectively. The joint distribution of the observation vector $Y = ( \, Y(\mt{s}_1), Y(\mt{s}_2), \ldots, Y(\mt{s}_D)\,)$ is a multinomial distribution with parameters $(T, \Pi_f)$ where $\Pi_f$ is a probability mass function indexed by $\mt{s}_i$ with
\begin{equation}
\Pi_f(\mt{s}_i) = \Pr\nolimits_f\{ \xi_1 \in \mt{S}_i \} = \int_{\mt{S}_i} \, f(\mt{s}) \,\mbox{d}\mathtt{s},
\label{model1}
\end{equation}
for $i=1,2,\ldots, D$. Here $\mbox{d}\mathtt{s}$ indicates integration with respect to the Lebesgue measure on $\mathcal{D}$. Note that $\Pi_f$ is a probability mass function over $\mathcal{G}$ which is the \textit{finest} parametrization of the unknown density $f$, that is, given the discretization $f$ cannot be recovered at resolutions finer than the pixel level averages. 

The model given by \eqref{model1} is slightly different from the traditional function estimation model with additive Gaussian white noise  \cite{Chu,Donoho2,Qiu2}. The common variance of the errors determines the precision of the image. The formulation considered here also adheres to the principles of digital image processing \cite{Pratt}. The total number of pixels $M^2$ is indicative of the designed precision of the imaging device. For satellite imagery in different spectral channels $M$ relates to the gridding of the target area on the ground for which image is supposed to be sufficiently precise (such as $3 \times 3$m$^2$ or $6 \times 6$m$^2$). The more precise the technology the larger value is assumed by $M$ (also lesser bias in the estimation of $f$). On the other hand the total number of samples $T$ relates to other noisy disturbance present in the channel. The lesser the noise the larger would be the value of $T$. We discuss it in more detail in section 2. From asymptotic considerations it has been shown that both formulations lead to equivalent sequence of experiments in the sense of Le Cam's measure of deficiency \cite{Brown,Nussbaum}. There is one glitch though. The drift function in the additive version of the density estimation problem becomes $f^{1/2}$.
\subsection{Semi-parametric mixture model for $\boldsymbol{f}$}
Towards this let $N(\mt{s}) \subset \mathcal{D}$ denote an open square with center $\mt{s} \in \mathcal{D}$. These open squares form a basis for the topology of the interior $(0,1) \times (0,1)$. Because the boundaries are assumed to have measure zero (as we assume the existence of density under Lebesgue measure) it is enough to restrict to the Borel $\sigma-$field of int$(\mathcal{D})$ for approximating image intensity function $f$. The extension to the whole of $\mathcal{D}$ is straightforward by treating $\mathcal{D}$ as a topological subspace of $\mathbb{R}^2$ \cite{Munkres}. Let $\mathcal{O}$ denote the collection of open squares of the form $N(\mt{s})$.

For any finite collection $N_1, N_2, \ldots, N_k \in \mathcal{O}$ and any compact set $A \subset \cup_i N_i$, let $ \{ (\rho_1, N_1),$ $ (\rho_2, N_2),\ldots, (\rho_k, N_k) \}$ be a \textit{smooth partition of unity} subordinate to the collection of open neighborhoods $N_i$'s (\textit{cf.} Theorem 3-11, \cite{Spivak} for its existence and properties). If the context is clear a partition of unity will be denoted by $(\rho_1, \rho_2, \ldots, \rho_k)$ to keep notations simple. Specifically, $\rho_i$'s are kernel-like smooth functions defined on $\mathcal{D}$ with $\rho_i$ vanishing outside $N_i$. The following properties are very useful.
\begin{align}\label{partition}
&\mbox{supp }\rho_i \subset N_i,& \nonumber \\
&0 \leq \rho_1(\mt{s})+ \rho_2(\mt{s})+ \ldots + \rho_k(\mt{s}) \leq 1 \mbox{ for all } \mt{s}\in \mathcal{D},  \\
&\rho_1(\mt{s})+ \rho_2(\mt{s})+ \ldots + \rho_k(\mt{s}) = 1 \mbox{ for all } \mt{s} \in A. \nonumber
\end{align} 

Define $\rho_0(\mt{s}) = 1 - \sum_i^k \rho_i(\mt{s})$, for all $\mt{s}\in \cal{D}$. Then $\rho_0(\mt{s})+\rho_1(\mt{s}) + \ldots + \rho_k(\mt{s}) \equiv 1 $ for every $\mt{s} \in \cal{D}$. Note that $\rho_0$ assumes the value $0$ on the compact set $A$. For the purpose of implementation any function defined on a continuous domain will be approximated by its average value at the finest resolution, that is, at the pixels $S_i$'s. 

The idea behind partition of unity is to decompose the domain of function into regions based on contrasting properties like local irregularities (such as high local oscillations or H\"{o}lder continuity of index $< 1$) and global smoothness (such as uniform twice differentiability). Edges tend to appear as boundaries between these contrasting regions. For this purpose indicator type partitions are not suitable as the artifice jeopardize the properties we are looking for. On the other hand smooth partitions of unity do not disturb the smoothness profile. However we cannot make a crisp judgement regarding edges and boundaries (even without noise) as $\rho_i \rho_j \equiv 0$ is not satisfied for $i \neq j$. More usefulness of smooth partitions of unity and other issues will be discussed in later sections.

Given any non-negative, integrable function $\rho$ defined on $\mathcal{D}$ let $\lambda_{\rho}$ denote the measure on $( \mathcal{D},\mathcal{B} )$
\begin{equation}\label{abscont}
\lambda_{\rho}(A) = \int_A \rho(\mt{s})\, \mbox{d}\mt{s} \quad \mbox{ for } A \in \mathcal{B},
\end{equation} 
where $\mathcal{B}$ is the Borel $\sigma$-field. As has already been observed we practically work on a finite sub-field of $\mathcal{B}$ generated by arbitrary unions of $S_i$'s, mostly rectangles of sub-images. If for another non-negative, integrable $g$, $\lambda_g$ is absolutely continuous with respect to $\lambda_{\rho}$ then it will be  abbreviated as $g \, {\scriptstyle \ll} \,\rho$ (see \cite{Billingsley}). For any density $f$ and a partition of unity $(\rho_1, \rho_2, \ldots, \rho_k)$ the density can be localized into $k$ pieces given by $\rho_i f$, $i=1,2,\ldots, k$. Notice that $\rho_i f {\scriptstyle \ll} \rho_i$. In case we can localize $f$ appropriately by making a right choice of the partition of unity functions so that the the neighborhoods dominating the $\rho_i$'s cover points of discontinuity of $f$ or $f^{\prime}$. The remaining component $\rho_0 f$ will extract the regular or smooth part of $f$ from the data. On the basis of this basic principle we choose the semi-parametric model for the unknown density $f$ by
\begin{equation}\label{model}
\mathcal{F} = \left\lbrace \sum_{i=1}^k \alpha_if_i(\mt{s}| \theta_i) + \alpha_0 g(\mt{s}): f_i(\cdot|\theta_i)\, {\scriptstyle \ll} \,\rho_i, g \,{\scriptstyle \ll} \,\rho_0, \mbox{ for some } (\rho_i), \sum_0^k \alpha_i =1  \right\rbrace.
\end{equation}

The model given by \eqref{model} is an extension of usual mixture model due to the non-parametric component $g$ and has an independent theoretical interest on its own. In this article \eqref{model} is adapted for certain specific application on edge-preservation and image restoration in mind. For comprehensive discussion on mixture models in classical statistics we refer to \cite{Lindsay}. Here for each $i$, $f_i(\mt{s}|\theta_i)$ denote family of locally defined parametrized densities which play a crucial role in fitting edges. In this paper, we develop a statistical methodology for edge detection and estimation by modeling the local features using a \textit{Local Template Model (LTM)} based on low-dimensional exponential family where the sufficient statistics allow different types of edges. The dimension is kept low with computational complexities in mind. The image domain $\mathcal{D}$ is scanned by fitting the LTM over pixel windows of fixed size (we reported the results with $11 \times 11$ windows) and significant locations were captured by testing presence of edge in multiple locations using the Holms procedure \cite{Holms,TSH}. The structure of \eqref{model} will be elaborated further in ensuing sections. The non-parametric component $g$ in \eqref{model} provides information about the smooth part or the background luminosity of the image. There are various smoothing techniques such as, kernel, local polynomials,  splines or other variational optimization principles \cite{Chambolle,Efrom4,Fan,Grace} available in the literature for estimating functions of uniform smoothness. The main challenge here is extraction of the component in presence of local irregularities generated by embedded objects. This is reflected in the complementary portion of \eqref{model} which is represented by a localized mixture model. In the present article, after eliminating local irregularities by maximum likelihood based optimization, we transform the problem to spectral domain (by virtue of bi-variate Fourier transform) and estimate the non-parametric component $g$. We minimize the squared $\ell_2$-distance between the empirical and theoretical Fourier coefficients  using Lagrangian relaxation by \textit{Thin Plate Spline (TPS)}. Spectral optimization turns out to be particularly simple and provides a closed form expression in the form of a kernel convolution with nonlinear bandwidth (that is, not a simple scaling) determined by the penalty parameters re-transformed using the inverse Fourier transform back to the spatial domain. In the one-dimensional case theoretical properties such as computation of mean integrated squared error (MISE), minimaxity and other optimality and implementation issues such as plug-in, cross-validation or thresholding properties  have been studied in detail. See for example, \cite{Efrom4,Hardle}. In our experiments with some of the benchmark images the LTM based edge preservation and TPS relaxation give fairly accurate outputs even under noisy condition without any Gibb's like phenomenon or increased bias. The plug-in method used here to estimate undetermined Lagrangian parameters  provides an alternative to Efromovich-Pinsker Block shrinkage method for regularizing Fourier coefficients developed in \cite{Efrom1,Efrom2}.

\subsection{Other methods: wavelets and Bayesian}
There are two other frequently used methodology for edge detection and edge preserving image restoration in noisy cases. They are wavelets and Bayesian image modeling. We discuss some salient features, commonalities and differences with the proposed method very briefly. Wavelet based recovery from noisy signals have been extensively studied for this problem with significant amount of success \cite{Anton,Donoho2,Lazzaro}. A comprehensive theoretical treatment of basic variational regularization of wavelets can be found in \cite{Chambolle}. For wavelet based methods optimal approximation depends critically on the space of functions being interpolated. Typically, smooth spaces involve Sobolev spaces and Besov spaces include more badly behaved functions (which are of interest here). Functions in these spaces can be well approximated by wavelets. However, the success of wavelet scheme critically depend on the filter banks and thresholding to attain consistency in noisy image without losing much in the interpolation problem without noise. Several thresholding schemes have been developed and sharpened over the years. Some of the basic issues are described in \cite{Donoho0}. We also demonstrate below in Figure 1, how the benchmark functions of \cite{Donoho0} can be constructed using the main model given by \eqref{model}. As far as similarity goes the wavelet functions due to their location shifts and scaling at different resolutions are similar to searching the image domain $\mathcal{D}$ by \textit{local} templates. Although orthogonality is a bonus in terms of computation it is not possible to locally deform wavelets using local parametric templates retaining their orthogonality. Several improvements and newer ideas have been added to achieve this flexibility in the form of Gabor filter banks, kernel deformation, curvelets, contourlets among others \cite{Candes,Do,Perona}. 

\begin{figure}[!h]
\centering
\includegraphics[scale=0.55]{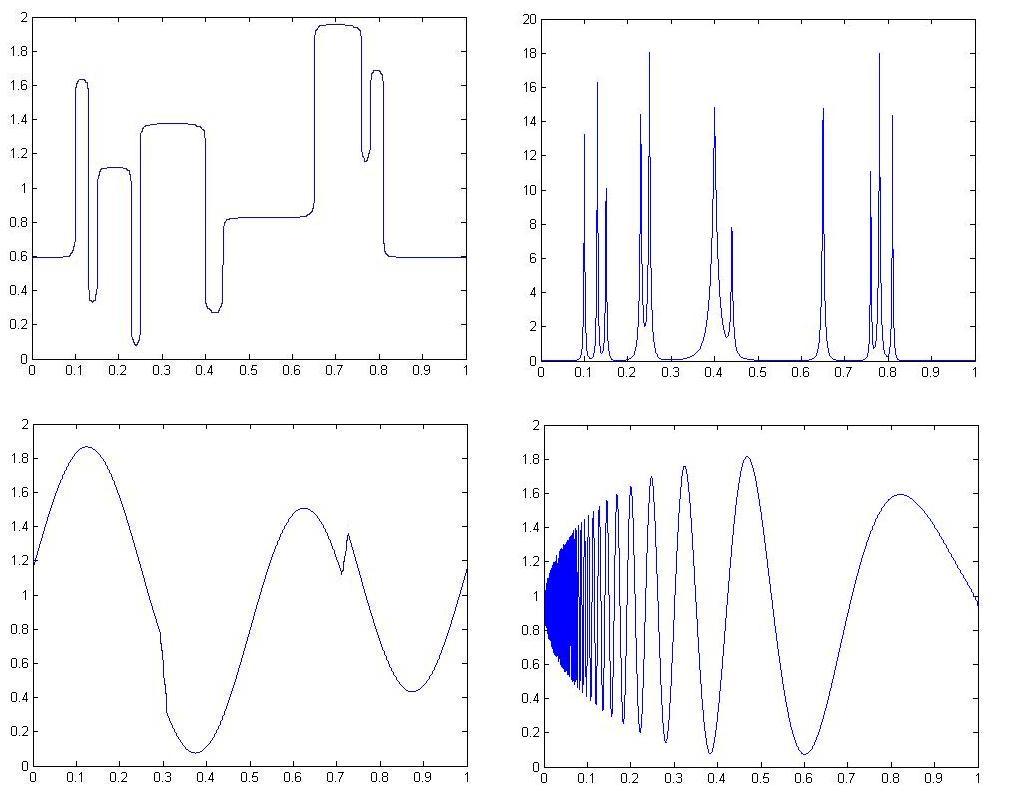}
\caption{Densities similar to the benchmark functions of \cite{Donoho0} obtained through the model \eqref{model}}
\end{figure}

In the Bayesian image processing literature basic attempt is to model the pixel probability vector $\Pi$ in \eqref{model1} using Markov random field priors \cite{Clifford,Geman0,Geman}. The key idea is to compute the posterior of the image, however a computational variant called Maximum a Posteriori Probability (MAP) estimate is commonly used. The method has generalized to variational Bayesian learning (ensemble learning) and Bayesian networks using extensive use of MCMC and other sampling methods \cite{Bernardo}. The construction of the main object of interest, that is, the prior for $\Pi$ is achieved by defining a suitable graphical structure of neighborhoods on the set $\mathcal{G}$ and then use the cliques of the graph to define the so called Gibbs distribution \cite{Clifford}. The prior constructed in this manner can be used as the relaxation for the maximum likelihood variational problem. This poses a challenging question: how to construct suitable Markov random field on $\mathcal{F}$ in \eqref{model}?     

In recent years, numerous smoothing procedures have been suggested in the statistical literature for edge-preserving function estimation from noisy data. General description can be found in the following references \cite{Fan,Hardle,Muller,Qiu2,Grace}.  Several 1D-methodology, have been discussed in literature. See references in \cite{Qiu} for details. A good overview of edge detection techniques has been discussed in \cite{Ziou}. There are many edge detection algorithms available in literature such as ISEF, Canny, Marr-Hildreth, Sobel, Kirsch, Lapla1 and Lapla2. A detailed and comparative study of these techniques can be found in \cite{Sharifi}. Reviews of the state of the art techniques for edge and line oriented approaches to contour detection can be seen in \cite{Papari}. Other edge preserving approaches in image processing have been discussed in literature \cite{Chu,Hillebrand} and references therin. 

The rest of the paper is organized as follows. In section 2 we provide theoretical details and some statistical issues related to our problem. Section 3 describes the methodology which we follow. Section 4 is dedicated for the application of our methodology on Lenna's image. In Section 5 we analyze the errors and conclude in Section 6 with some limitations, modifications and possible extensions. Few rigorous theoretical proofs and the detailed algorithm can be seen in the Appendix.
\section{Theory basics}\label{theory}
\subsection{Basic distribution theory}
Consider the multinomial model \eqref{model1} and the empirical histogram descriptor defined on $\mathcal{S} := \cup_{i=1}^D\, \mt{S}_i\,\subset \mathcal{D}$. The almost everywhere definition suffices as $\lambda(\mathcal{D} \backslash \mathcal{S}) = 0$ (see section \ref{intro}).   
\begin{equation}\label{estimator}
f^{\textit{emp}}_T (\mt{s}) = \frac{D}{T}\sum_{i=1}^D \, Y(\mt{s}_i)\, I(\mt{s} \in S_i),
\end{equation}    
where $I(\cdot)$ denotes the indicator function. The expected \textit{image density}, the main parameter of interest under the model, is given by  
\begin{equation}\label{mean:est}
\pi_f (\mt{s}) = D\sum_{i=1}^D \, \Pi_f(\mt{s}_i) \, I(\mt{s} \in S_i).
\end{equation} 
Note that $\pi_f$ is provides the information on average intensity of the image at the pixels and thus represents the true image after \textit{theoretical} noise removal. Under model assumptions (be it parametric or non-parametric) the theoretical expectation functional  $\mathbb{E}_{\pi_f}$ is assumed to be unknown. The central goal of any statistical procedure is to reconstruct this functional from pixel site data $Y$ and structural model assumptions involving unknown parameters and functions (in our case it is provided by \eqref{model}). The empirical evidence is summarized through $\femp$. If we estimate the target functional $\mathbb{E}_{\pi_f}$ simply by $\mathbb{E}_{\mfemp}$ we get an under-smoothed estimate missing out relevant substructures in noisy images and becoming extremely sensitive to minor degradation and noise, losing \textit{robustness} altogether.    

With the above scaling for any integrable function $\rho$ defined on $\mathcal{D}$, in view of \eqref{model1},\eqref{estimator} and \eqref{mean:est} the expectation projection functionals can be calculated as follows. 
\begin{equation}\label{proj_identity}
\begin{split}
& \mathbb{E}_{\femp}( \rho) = \frac{1}{T} \sum_{i=1}^D\, \bar{\rho}(\mt{s}_i) Y(\mt{s}_i)  \\ 
& \mathbb{E}_{\pi_f} (\rho) = \sum_{i=1}^D\, \bar{\rho}(\mt{s}_i) \Pi_f(\mt{s}_i),
\end{split}
\end{equation}
where $\bar{\rho}(\mt{s}_i) =D \int_{\mt{S}_i}\, \rho(\mt{s})\,\ds = \left(\int_{\mt{S}_i}\, \rho(\mt{s})\,\ds\right)/\left(\int_{\mt{S}_i}\! \ds\right)$ for $i=1,2, \ldots, D$. We interpret $\bar{\rho}$ as the image projection of original $\rho$. The quantities on the right hand side of \eqref{proj_identity} can be interpreted as discretized expectation functionals at the pixel sites $\mathcal{G}$
. The basic projection identity helps us understanding how the class of estimators defined by linear functionals of densities defined on $(\mathcal{D}, \mathcal{B})$ reduce to linear functions of sufficient statistics $\{Y(\mt{s}_i)\}$ under the multinomial model \eqref{model1}.

In this section we obtain some heuristic approximations which can be made into precise asymptotics if so desired. We also suppress the suffix $T$ wherever possible for notational simplicity (both $T$ and $D= M^2$ are given quantities). However we shall keep the fact in mind that both $D$ and $T$ are required to be large for greater precision in estimation. For any $m \times m$ square neighborhood $N$  define its \textit{spatial bandwidth} to be $h_m= (m/K)$. For a rectangular $m \times r$ neighborhood its spatial bandwidth will be $(h_m h_r)^{1/2}$. We implicitly assume the following in what follows. For any rectangular $m \times r$ neighborhood  
\begin{equation}\label{assumption}
a\, h_m h_r \leq \Pi_f(N) \leq A \, h_m h_r, 
\end{equation}
for suitable global constants $a, A > 0$.   

A function $\rho$ defined on $\mathcal{D}$ is \textit{image measurable} if it assumes constant values over pixels $\mt{S}_i$'s. For an image measurable $\rho$, it can be easily seen  $\rho = \bar{\rho}$. We list some useful properties of multinomial counts in the following proposition. 
\begin{prop}\label{moment} 
\begin{enumerate}\item[(i)] Let $\rho$ be any image measurable function defined on $\mathcal{D}$. Then 
\begin{equation}\label{prop1}
\begin{split}
(a) \,\mathbb{E}_f \mathbb{E}_{\femp} (\rho) &= \mathbb{E}_{\pi_f} (\rho) \triangleq \mathbb{E}_{\Pi_f} (\rho) \phantom{\mathop{\frac{1}{T} \left[ \mathbb{E}(\rho^2\,|\Pi_f) -  \mathbb{E}^2(\rho\,|\Pi_f) \right] }}\\
\vspace{3mm} 
(b) \mbox{Var}_f \left( \mathbb{E}_{\femp} (\rho)\right) &= \frac{1}{T}\,\mbox{Var}(\rho\,| \pi_f) := \frac{1}{T} \left[ \mathbb{E}_{\pi_f}(\rho^2) -  \mathbb{E}_{\pi_f}^2(\rho) \right] \\
\vspace{5mm}
 &\triangleq \frac{1}{T}\mbox{Var}(\rho| \Pi_f).
\end{split}  
\end{equation}
\item[(ii)] Let $N_i \subset \mathcal{D}$ be $m \times m$ pixel neighborhoods and $\rho_i$ be image measurable function vanishing outside $N_i$ for i=1,2. Then
\begin{equation} \begin{split}
\mbox{Cov}_f \,\left( \mathbb{E}_{\femp}(\rho_1),\mathbb{E}_{\femp}(\rho_2) \right) & 
= \frac{1}{T}\mbox{Cov}(\rho_1, \rho_2 | \pi_f) =\frac{1}{T} \left( \mathbb{E}_{\pi_f}(\rho_1 \rho_2)  - \mathbb{E}_{\pi_f}(\rho_1) \, \mathbb{E}_{\pi_f}(\rho_2) \right) \\
&\triangleq \frac{1}{T} \mbox{Cov}(\rho_1, \rho_2 \, | \Pi_f) \\
&= \frac{1}{T} \lbrace \Pi_f(N_1N_2) \,\mathbb{E}_{\Pi_f} (\rho_1 \rho_2 \cond I_{N_1}=I_{N_2}=1)  \\ 
\vspace{5mm}
 &\quad  - \Pi_f(N_1)\, \Pi_f(N_2) \, \mathbb{E}_{\Pi_f}(\rho_1 \cond I_{N_1}=1) \,\mathbb{E}_{\Pi_f}(\rho_2 \cond I_{N_2}=1) \rbrace, 
\end{split} \end{equation}
where $I_N$ is the indicator variable attached to the neighborhood $N$ and $\mathbb{E}(\cdot \cond \cdot)$ is the generic notation for conditional expectation of the first argument given the second. In the last line the conditional expectations are calculated under the discrete probability setup $(\mathcal{G}, \Pi_f)$ with the random variables $\rho_1, \rho_2: \mathcal{G} \rightarrow \mathbb{R}$ defined by their image measurability.
\end{enumerate}
\end{prop}
 
\subsection{Discretization of MISE}
The mean integrated squared error (MISE) has become a benchmark risk function in density estimation over the decades in statistical literature \cite{Efrom4,Fan,Hardle,Silverman}. The special structure of MISE has its natural roots for the class of kernel smoothed density estimators and clarity of bias-variance trade-off. The ANOVA decomposition of squared error and its intrinsic connection with Euclidean geometry is another attraction. Due to squaring of the densities it gets less influenced by improbable values in comparison with $L_1$ or  Hellinger metrics.  It is known that the $L_1$ error is equivalent to the total variation distance which represents the worst case discrepancy in predicting events. However this makes the metric more suitable for other applications. Although various other error assessment of non-parametric procedures has become more popular in recent years particularly in areas of machine learning such as, unsupervised learning, classification, prediction, data mining, and support vector machines we choose the MISE as the risk function for the purpose of the present article \cite{Hastie,Vapnik}. 

We proceed with a useful discretization lemma below which gives us a significant reduction in the class of estimators under the multinomial model \eqref{model1}. In the language of statistical decision theory the lemma tells us while summarizing the data by the empirical histogram $\femp$, histogram type smooths form a \textit{complete class} under MISE \cite{Berger}. We require a slight variation of the usual conditional version of the Jensen's inequality for this purpose and later. In order to keep presentation neat we state it separately as a lemma. The Jensen's inequality in its conditional form is essentially the basis of \textit{Rao-Blackwellization} in statistics. For a proof see  \cite{Billingsley}. 
\begin{lem}\label{jensen} Suppose $g: L \times L \rightarrow \mathbb{R}$ is a continuous function defined on a suitable interval $ L \subseteq \mathbb{R}$, such that $g(\cdot, t)$ is convex for every $t \in L$.  Consider a random variable $U$ and a random object $V$ defined on the same probability space with with $U$ having finite second moment. Assume that $U$ takes values in $L$. Then
\begin{equation}\label{jensen_ineq}
\mathbb{E}\lbrace g (U, \tau(V) ) \cond V \rbrace \geq  g\!\left(\mathbb{E}(U \cond V),\tau(V) \right) \quad \mbox{ \textit{a.e.} }
\end{equation}
where $\tau$ is a $L$-valued function defined on the range of $V$. If $g(\cdot, t)$ is strictly convex for every $t \in L$, equality holds if and only if $U$ is a function of $V$.
\end{lem}
The \textit{random object} part may be made technically more precise if we recognize it as a sub $\sigma$-field and $\tau$ being measurable with respect to that  sub $\sigma$-field. However, we discuss only the heuristics here.
  
Consider an arbitrary estimator $\hat{f}_T(\mt{s};Y)$ of the image density $\pi_f(\mt{s})$. Under the model \eqref{model} $\hat{f}_T(\cdot; Y)$ takes values in $\mathcal{F}$. For the particular estimator the MISE risk function is defined by
\begin{equation}\label{mise_def}
\begin{split}
R(\hat{f}_T, f) &= \mathbb{E}_{f} \int_{\mathcal{D}}\, \left( \hat{f}_T(\mt{s}; Y) - \mathbb{E}_f \femp(\mt{s}) \right)^2\ds \\
& = \mathbb{E}_{f} \int_{\mathcal{D}}\, \left( \hat{f}_T(\mt{s}; Y) - \pi_f(\mt{s}) \right)^2\ds \\
& = \sum_{i=1}^D\, \mathbb{E}_f \int_{\mt{S}_i}\,\left( \hat{f}_T(\mt{s}; Y) - D\Pi_f(\mt{s}_i) \right)^2\ds.
\end{split}
\end{equation}   
Next conceive of a extraneous randomization $(I, X)$ which is realized by choosing a random index $1\leq I \leq D$ and then a random point $X$ from the uniform distribution on $\mt{S}_I$. The data vector $Y$ and the other random object. All these random objects can be defined on a suitable product probability space with expectation denoted by $\mathbb{E}_f^*$. We apply Lemma \eqref{jensen_ineq} with $U=X$, $V=(Y, I)$ and $g(x) = x^2$, to obtain the following lower bound.  
\begin{equation}\label{mise_RB}
\begin{split}
R(\hat{f}_T, f) & = D^2 \mathbb{E}_f^* \left( \hat{f}_T(X; Y) - D\Pi_f(\mt{s}_I) \right)^2 \\
& \geq D^2 \mathbb{E}_f^*\, \left( \mathbb{E}^*\lbrace \hat{f}_T(X; Y) \cond V \rbrace - D\Pi_f(\mt{s}_I) \right)^2.  
\end{split}
\end{equation}   
The subscript $f$ is omitted in the inner expectation as it is conditional on the sufficient statistic $Y$, hence expectation is carried out only with respect to external randomization. It can be easily checked that for any $i$,  $\mathbb{E}_f^*\lbrace \hat{f}_T (X; Y) \cond I=i, Y\rbrace = \left( \int_{\mt{S}_i}\,\hat{f}_T(\mt{s}) \ds \right) / \left( \int_{\mt{S}_i} \ds \right) $ $= D \int_{\mt{S}_i}\,\hat{f}_T(\mt{s}) \ds$. 
In other words the Rao-Blackwellized estimator is constant on the elementary pixels $S_i$'s. Therefore we conclude for the estimation of $\pi_f(\mt{s}) = \mathbb{E}_f \femp(\mt{s})$ the class of \textit{image measurable histogram estimators} given by 
\begin{equation}\label{hist_est}
\mathcal{H} = \left\lbrace \hat{f}_T(\mt{s}) = \sum_{i=1}^D\, \hat{\pi}_i(Y) \, I(\mt{s} \in \mt{S}_i) : \hat{\pi}_i(Y) \geq 0, \mbox{ for all } i,  \sum_{i=1}^D\, \hat{\pi}_i(Y) = D, \mbox{\textit{a.e.}} \right\rbrace
\end{equation}
forms a complete class under the MISE given by \eqref{mise_def} (\textit{viz.} \cite{Berger}). Under the basic semi-parametric model \eqref{model} the edge-preserving reconstructions $\hat{f}_T\in \mathcal{F}$. Unfortunately the LTM based model does not contain histogram like functions. However histograms form a dense subset in a larger space of densities (\textit{viz.} \cite{Brown,Donoho0,Nussbaum}). In order to rectify this problem we modify the continuous version of MISE and adapt it to the current context. Given a discretization (digitization) $\lbrace \mt{S}_i, i \in \mathcal{G}\rbrace$ of a continuous (analogue) problem The discretized mean square error (DMSE) for an estimator $\hat{f}_T$ under the data generating density $f$ is defined as
\begin{equation}\label{dmse_def}
\begin{split}
R_{D}( \hat{f}_T, f) & \triangleq \frac{1}{h_1} \,\mathbb{E}_f\sum_{i=1}^D\, \left( \int_{\mt{S}_i}\, [\hat{f}_T(\mt{s}) - f(\mt{s})] \ds \right)^{\!2}.\\
                             & = \frac{1}{h_1} \,\mathbb{E}_f\sum_{i=1}^D\, [ \hat{\Pi}_T(\mt{s}_i) - \Pi_f(\mt{s}_i) ]^2,   
\end{split}
\end{equation}    
where $\hat{\Pi}_T(\mt{s}_i) = \int_{\mt{S}_i}\, \hat{f}_T(\mt{s}) \ds$. The validity of the approximation in \eqref{dmse_def} can be naturally justified as a Riemann sum approximation of the MISE. Also, by virtue of the Rao-Blackwellization step in \eqref{mise_RB} DMSE provides a lower bound of MISE for any given partition $\lbrace \mt{S}_i, i \in \mathcal{G}\rbrace$. This is a natural adaptation of classical MISE under discretization and is defined for continuous estimates and data generating densities. It may also be seen that the class of histogram estimators defined by \eqref{hist_est} remains  \textit{essentially complete} or \textit{risk equivalent} (\textit{viz.} \cite{Berger}) with the class of general smooth estimators on the continuous (analogue) domain. The optimal rate for bivariate density estimation with kernel smoothing can be found in \cite{Wand}. However, it should be noted that the underlying assumptions are smooth density models. The optimum rate achievable under wavelets for vrious function classes are described by \cite{Donoho0,Donoho}. However, in present situation without assuming any further smoothness assumptions on the set of edges it is virtually impossible to conclude about rate. Some theoretical approach towards optimal rates will be taken up elsewhere.  

\subsection{Variational optimization of the semi-parametric likelihood with TPS Lagrangian relaxation}
After a suitable reparametrization  $ f \in \mathcal{F}$ can be represented as $f(\mt{s}) = \rho_0(\mt{s}) \,g_0(\mt{s}) + \sum_{i=1}^k \rho_i(\mt{s})\, g_i(\mt{s}|\theta_i)$ with respect to the associated partition of unity $(\rho_i)$ where $g_i(\mt{s}|\theta_i)$'s are defined by
\begin{equation}\label{LTM1}
\int_A \,g_i(\mt{s}|\theta_i)\,\rho_i(\mt{s}) \ds = \alpha_i \int_A \,f_i(\mt{s}|\theta_i) \ds, 
\end{equation} 
for measurable sets $A \in \mathcal{B}$ and $i=0, 1, \ldots, k$. Because we are dealing with probability distributions we also get  
$$ \int_{\mathcal{D}}\, \rho_0(\mt{s})\, g_0(\mt{s}) \ds + \sum_{i=1}^k\, \int_{\mathcal{D}}\, \rho_i(\mt{s})\, g_i(\mt{s}|\theta_i) \ds = 1.$$ 
This is obtained by invoking the absolute continuity conditions in \eqref{model}. This requires an application of Radon-Nikodym theorem \cite{Billingsley}. The variational optimization problem we propose does not directly optimize the DMSE. In stead we consider the following upper bound to arrive at a more intuitive and computable criterion using separation of variables. The main unstructured part of the procedure is determination of the partitions of unity $(\rho_i, N_i)$, for $ i = 1, 2, \dotsi, k$ where k is also assumed to be unknown unknown. This step is achieved by scanning the image with a fixed rectangular neighborhood. At each placement a decision problem is posed whether the the local parametric template with discontinuity should be fitted or not. The decision making is not done independently of the placement, rather a multiple testing formulation has been considered to detect the significant neighborhood whose union covers the unknown subset of discontinuities. Once $\rho_i$'s are selected the remaining parameter (conditioned on the choice of partitions of unity) are $\boldsymbol{\theta}=(\theta_1,\theta_2,\dotsi,\theta_k)$ and a smooth density estimator $g_0(\mt{s})$ of $\rho_0(\mt{s}) \femp (\mt{s})$. We give an intuitve justification for separation of variables using MISE. Although the results are also true for DMSE we omit that for notational inconvenience.    
\begin{align}\label{mise}
R(\hat{f}_T, f) &= \mathbb{E}_f \int_{\mathcal{D}}\, (\hat{f}_T(\mt{s})- f(\mt{s}))^2 \ds \nonumber \\
                &= \mathbb{E}_f \int_{\mathcal{D}}\, \left(\hat{f}_T(\mt{s})- \sum_{i=1}^k \,\alpha_if_i(\mt{s}| \theta_i) - \alpha_0 f_0(\mt{s})\right)^2 \ds  \nonumber \\
                &= \mathbb{E}_f \int_{\mathcal{D}}\, \left( \rho_0(\mt{s})\, (\hat{f}_T(\mt{s})-g_0(\mt{s}))+ \sum_{i=1}^k \,\rho_i(\mt{s})\, (\hat{f}_T(\mt{s})-g_i(\mt{s}|\theta_i)) \right)^2 \ds \nonumber \\
                & \leq \mathbb{E}_f \int_{\mathcal{D}}\, \rho_0(\mt{s})\, \left(\hat{f}_T(\mt{s})-g_0(\mt{s}) \right)^2 \ds \,+\, \sum_{i=1}^k \,\mathbb{E}_f \int_{\mathcal{D}}\, \rho_i(\mt{s})\, \left(\hat{f}_T(\mt{s})-g_i(\mt{s}| \theta_i) \right)^2 \ds.
\end{align}

The last inequality in \eqref{mise} follows from Jensen's inequality for each $\mt{s}$. This inequality gives us an useful upper bound for the original MISE. The main advantage is that the original mixture problem is split into $k$ local approximation problems and a global smoothness problem (separation of variables). It also gives us a preference over the choice of $(\rho_i)$ (and the dominating neighborhoods). Note the the upper bound would be exact if $\rho_i$'s satisfied the point-wise orthogonality condition $\rho_k\,\rho_l \equiv 0$ for $k \neq l$. This is not achievable due to smoothness of the partitions. The inequality \eqref{mise} suggests that the covering should be as tight as possible to make the bound close to the exact MISE. Note that for $ 1 \leq i \geq k$ the variational problem is parametric hence the natural criterion is log-likelihood. Therefore we replace the MISE part by the likelihood criterion in \eqref{mise} to obtain the following minimization problem in $(\boldsymbol{\theta}, g_0)$. 
\begin{equation}\label{optimization}
\begin{split}
 Q(\boldsymbol{\theta},g_0|\femp) & \triangleq - \sum_{i=1}^k\, \int_{\mathcal{D}}\, \rho_i(\mt{s})\,\log g_i(\mt{s}|\theta_i) \femp(\mt{s}) \ds  \\
& \quad + \int_{\mathcal{D}}\, (\rho_0(\mt{s}) \femp(\mt{s})- g_0(\mt{s}))^2 \,\ds.
\end{split}
\end{equation}  

Note that as if expected with unconstrained variational problems, without any regularization the problem has a trivial solution, namely, $\rho_0 \equiv 1$, $\rho_i \equiv 0$, for $i \geq 1$ and $\rho_0(\mt{s})\,\femp(\mt{s})- g_0(\mt{s})\equiv 0$. Therefore a regularization is required. For the present approach we choose the TPS regularization assuming $g_0$ to be twice continuously differentiable. Let $D_2$ denote the second derivative matrix (the Hessian) of a function of two variables. The TPS regularization penalizes for $||D_2f||^2= \int_{\mathcal{D}}\, \mathop{||} D_2 f (\mt{s})\mathop{||}_F^2  \ds$, where the norm inside the integration denotes the sum of squares of the entries of the Hessian (squared Frobenious norm). After adding a suitable Lagrangian penalty parameter the final variational optimization problem with TPS regularization becomes
\begin{equation}\label{main_opt}
(\hat{\boldsymbol{\theta}}, \hat{g}_0) = \mathop{\arg\min}
_{f \in \mathcal{F},\, g_0 \,{\scriptscriptstyle <\!\!<}\, \rho_0} \left\{ Q(\boldsymbol{\theta}, g_0 | \femp) \, + \, \lambda || D_2  g_0||^2 \right\}.
\end{equation}
For the parametric part the optimization the discretization can be done using Theorem \ref{jensen}. This reduces the value of the objective function. Since $(-\log x\,)$ is a convex function it can be readily verified that for each $i$,
$$ -\int_{\mathcal{D}}\, \rho_i(\mt{s})\,\log g_i(\mt{s}|\theta_i) \femp(\mt{s}) \ds \geq -\frac{1}{T}\,\sum_{i=1}^D\, \log \bar{g}(\mt{s_i})| \theta_i) \, Y(\mt{s}_i),
$$
where $\bar{g}(\mt{s_i}) = (\int_{\mt{S}_i}\, \rho_i(\mt{s}) \ds ) / ( \int_{S_i}\, \rho_i(\mt{s}) \, \ds)$ whenever the denominator is positive. 

One last remark about the absolute continuity of the nonparametric component $g_0$. We have found that the solution in this form is more stable as the solution for the unweighted spline $\hat{g}_0$ automatically satisfies $\hat{g}_0/ \rho_0$ bounded. If there is any indication of possible singularity, a convolution of the Fourier coefficients of $\hat{g}_0$ with Fourier coefficients of $\rho_0$ serves the purpose. However the near singularity ( to be decided some tolerance level) might indicate something more serious, namely, certain irregular local neighborhoods of the image not being captured due to type II error of multiple testing procedure.  

\subsection{Embedding in spectral domain}
In order to describe the image embedding to the spectral domain first we describe the embedding of one dimensional signals defined on a compact intervals to the spectral domain in some detail. For signals defined on two dimensional domains we take the tensor product of two one dimensional spectra. In this subsection we briefly develop the notations and mention some of the useful properties of spectral embedding of signals which will be helpful in understanding the methodology better. Denote the set of complex numbers by $\mathbb{C}$ and let $\mathbb{T} \subset \mathbb{C}$ be the unit circle on the complex plane, $\mathbb{T}=\{ z \in \mathbb{C} : |z|=1\}$. This is a compact abelian group under the complex multiplication and the complex conjugacy satisfying $\bar{z} = z^{-1}$ on $\mathbb{T}$. The unit circle is the basic construct for spectral analysis of signals. The natural map of the unit interval $[0,1]$ into $\mathbb{T}$ is the exponential map $ x \mapsto z= \exp (2\pi j x) \triangleq \cos (2 \pi x) + j \sin (2 \pi x)$, with $j^2 =-1$, defined by the Euler's formula. The Lebesgue measure on $[0,1]$ and the Haar measure on $\mathbb{T}$ are related by $2 \pi\, dx = dz$ (\textit{viz.} \cite{Rudin}) and moreover, for any integrable function $f:\mathbb{T} \rightarrow \mathbb{C}$ the problem of integration can be transformed to the real line by the following formula.
$$\int_{\mathbb{T}}\, f(z)\, dz  = 2 \pi \,\int_0^1\, f( e^{2\pi j x}) \, dx. 
$$
Note that $g(x) = f( e^{2\pi j x})$ extends as a periodic function on the real line with period $2\pi$. The continuous homomorphisms of $\mathbb{T}$ into itself (the \textit{characters}) are given by,
\begin{equation}\label{characters}
\hat{\mathbb{T}} = \lbrace \phi_n(z) \triangleq z \mapsto z^n: n \in \mathbb{Z} \rbrace,
\end{equation}
where $\mathbb{Z}$ denotes the set of integers and $z^0 \equiv 1$. The sequence of functions $\lbrace\phi_n(z)\rbrace$ are mutually orthogonal with respect to $dz$ and $\lbrace (2\pi)^{-1/2} \phi_n(z) : n \in \mathbb{Z} \rbrace$ forms a complete orthonormal basis of $L_2(\mathbb{T}, dz)$. The Fourier coefficients of a function $f \in L_2(\mathbb{T}, dz)$ are given by 
\begin{equation}\label{fourier_def}
\hat{f}_n = \frac{1}{2\pi}\, \int_{\mathbb{T}}\, z^{-n} \, f(z) \, dz, \quad \mbox{ for } n\in \mathbb{Z}.
\end{equation}
The fundamental isometry property of the transform yields the following Fourier inversion formula and preservation of inner products and distances in $L_2(\mathbb{T}, dz)$ (Plancherel theorem, Perseval's identity, \textit{viz.}  \cite{Rudin}). For $f, g \in L_2(\mathbb{T}, dz)$
\begin{equation}\label{fourier_inv}
\begin{split}
f(z) &= \sum_{n=-\infty}^{\infty}\, \hat{f}_n\, z^n \\
\frac{1}{2 \pi}\,\int_{\mathbb{T}}\,\, f(z)\, \overline{\mathop{g(z)}}\,\, dz &=\sum_{n=-\infty}^{\infty}\, \hat{f}_n \,\overline{\hat{g}}_n \\
\frac{1}{2 \pi}\,\int_{\mathbb{T}}\,\, |f(z) - g(z)|^{2}\, dz &= \sum_{n=-\infty}^{\infty}\, |\hat{f}_n - \hat{g}_n|^{2}  
\end{split}
\end{equation}
For differentiable functions $f:\mathbb{T} \rightarrow \mathbb{C}$ we shall denote the derivative \textit{along the circle} at a point $z=z_0$ either by $f_{\!z} (z_0)$ or $\mathop{\frac{\partial }{\partial z}f (z_0)}$ to distinguish it from the derivative in the complex plane $d/dz$. The definition being  as 
$$   f_{\!z}(z_0) \triangleq \mathop{\frac{\partial f}{\partial z}} (z_0) = 2 \pi j \, \lim_{ h \curvearrowright 1} \frac{ f(hz_0) - f(z_0)}{ h - 1},
$$
where $\curvearrowright$ indicates convergence along $\mathbb{T}$. If $f_{\!z} \in L_2(\mathbb{T}, dz)$ the Fourier coefficients satisfy
\begin{equation}\label{fourier_deriv}
\hat{f}_{\!z, n} = n \hat{f}_n, \quad \mbox{ for } n \in \mathbb{Z}.
\end{equation}
From previous discussions we have already observed that Fourier coefficient contains the full information present in the signals which can be recovered by the inverse Fourier transform described by \eqref{fourier_inv}. We shall state one more important property of Fourier transforms (which is more crucial for density estimation problems) before we finish this discussion. If we have non-negative signals such as histograms, $\femp$,  $\pi_f$ or $\hat{f}_T$ discussed  section 2.1 one has one very useful property of Fourier transforms, that is, Bochner-Herglotz theorem, which states that the Fourier coefficients $\lbrace \hat{f}_n: n \in \mathbb{Z}\rbrace$ is a non-negative definite sequence. Moreover the converse is also true if $\hat{f}_0 =1$. Therefore \textit{spectral embedding} of a one dimensional spatial signal $f \in L_2([0,1], dx)$ is given by the sequence $\hat{f} =\left(\hat{\tilde{f}}_n: n \in \mathbb{Z}\right) \in \ell_2(\mathbb{Z}) $ where $\tilde{f} \in L_2(\mathbb{T}, dz)$ is defined by $\tilde{f}(z) = f(x)$ for $z = \exp (2 \pi j x)$. There is some ambiguity in the definition at $z=1$, however that is a set of measure zero and any choice can be made.

Next we describe how to discretize of $\mathbb{T}$ using the orientation of the unit circle which can be implemented in a fairly simple manner. Suppose we want to subdivide the unit circle into $K$ arcs of equal length and put a pointer to the centers of the arcs as well. For that we consider the set of $2K$-th roots of unity, that is, $\mbt{\Omega} = \lbrace z: z^{2K}=1\rbrace$. Let $\omega_{2K} = \exp (\pi j/K)$ be the $2K$-th root of unity which is nearest to $1$ in terms of arc-length in the counter-clock wise direction. Then $\mbt{\Omega}$ can be enumerated as $\lbrace \omega_{2K}^m: m=0,1,\dotsi, (2K-1) \rbrace$. In this enumeration the subset given by $\mbt{\Omega}_{odd} = \lbrace \omega_{2K}^m: m \mbox{ odd }\rbrace$, consisting of $K$ elements point to the centers of the $K$ arc intervals $\mathcal{T}= \lbrace ( \omega_{2K}^{2m}, \omega_{2K}^{2m+2}): m=0, 1, \dotsi, (K-1) \rbrace$. Note that for the last interval the end point is $ \omega_{2K}^{2K}=1$. It is also interesting to note that the end points corresponding to the even powers of $\omega_{2K}$ are also $K$-th roots of unity. This is an important observation because this property is used non-trivially in computing the recursions for the fast Fourier transform (FFT) algorithm for the discrete Fourier transform (DFT).
  
Finally, the embedding of the spatial domain $\mathcal{D}$ of the image begin with the torus, $\mathbb{T}^2 = \mathbb{T} \times \mathbb{T}$. The homomorphisms of $\mathbb{T}^2$ are given by $(z,w) \mapsto (z^m, w^n)$ for $(m, n) \in \mathbb{Z} \times \mathbb{Z}$. All the properties of Fourier transform on $\mathbb{T}$ carries through for the bivariate case. Following the convention in section \ref{intro}, a generic point on $\mathbb{T}^2$ will be denoted by $\mbt{\omega} = (z, w)$. Once the spatial domain is embedded to the torus the pixel centers $\mt{s}_i$ is mapped to a point on the torus of the form $\mbt{\omega}_i = (\omega_{2M}^{2k + 1}, \omega_{2M}^{2l+1})$ for some $0 \leq k, l \leq (M-1)$. The pixel rectangles $\mt{S}_i$'s are smoothly mapped onto $\mathbb{T}^2$-rectangles of the form $\mbt{\Omega}_i = U \times V$ where $U, V \in \mathcal{T}$ for each $i=1,2,\dotsi, D$. The group structure in $\mathbb{T}^2$ is specified by $\mbt{\omega}_1\mbt{\omega}_2 = (z_1z_2,w_1w_2)$   
for $\mbt{\omega}_1, \mbt{\omega}_2 \in \mathbb{T}^2$. The conjugacy operation is $\overline{\mbt{\omega}}_1 = (\overline{z}_1,\overline{w}_1)$, and the inverse is $\mbt{\omega}_1^{-1} = \overline{\mbt{\omega}}_1$. The natural metric in this group is defined by
\begin{equation}
\|\mbt{\omega}_1\overline{\mbt{\omega}}_2\| \triangleq \sqrt{|z_1\overline{z}_2 -1|^2 + |w_1\overline{w}_2 -1|^2}.
\end{equation}
Finally let open balls centred at point $\mbt{\omega}$ with radius $r$ be denoted by $B_r({\mbt{\omega}})$.

\section{Methodology}
In this section we describe the detailed methodology of finding the edge-preserving smooth density estimate from the class $\mathcal{F}$ defined in \eqref{model}. We begin by describing the \textit{Local Template Model} (LTM), for testing the presence of an edge in the local region of interest. We shift the LTM over the entire image, keeping track of the $p$-values for the tests. Multiple testing is performed to determine the regions $N_i$ having the edges at level $\alpha$. After obtaining the $N_i$'s, a linear programming problem is solved to obtain, the optimal partition of unity. The local regions are estimated through the LTM, while the density estimate for the remaining smooth region is then obtained through the TPS regularization using the Fourier basis. 

\subsection{Local Template Model (LTM)}
Since we have explained the Fourier technique in the spectral domain in section 2.4, for sake of uniformity in notation, we explain this parametric model in the spectral domain as well. As mentioned in the previous sections, this model is used for both edge detection and as well as local density estimation. Let $\tilde{N}(\mbt{\omega}_0)$ denote the an open neighborhood of size $m\times m$ pixels $(m = 2t + 1)$ in $\mathbb{T}^2$ having center at $\mbt{\omega}_0$ . $\tilde{N}(\mbt{\omega}_0)$ can be thought of as the image of an analogous planar neighborhood $N(\mt{s}_0)$ which can explicitly written as
\[
N(\mt{s}_0) = \left(\frac{i_0 - t}{M}, \frac{i_0 + t + 1}{M}\right) \times \left(\frac{j_0 - t}{M}, \frac{j_0 + t + 1}{M}\right)
\]
where $\mt{s}_0 = \left(\frac{i_0 + 1/2}{M},\frac{j_0 + 1/2}{M}\right)$. Note that we can find $r_1,r_2 > 0$ such that $B_{r_1}(\mbt{\omega}_0) \subset \tilde{N}(\mbt{\omega}_0) \subset B_{r_2}(\mbt{\omega}_0)$ for the continuous domain. 

Let $\rho(\cdot)$ denote a smooth non-negative function with support($\rho) \subset \tilde{N}(\boldsymbol{\mathtt{\omega}}_0)$, and 
$\mathbb{E}(\rho) > 0$ with respect to the Haar measure ($\domg = \mbox{d}z \times \mbox{d}w$) on $\mathbb{T}^2$. 
Since the Lebesgue measure does not change by the spectral transformation, we can define a measure following \eqref{abscont} on $(\mathbb{T}^2,\mathcal{B}_{\mathbb{T}^2})$, (where $\mathcal{B}_{\mathbb{T}^2}$ is the Borel $\sigma$-field on $\mathbb{T}^2$) as
\begin{equation}
\tilde{\lambda}_{\rho}(B) = \int_B \rho(\boldsymbol{\mathtt{\omega}})\, \mbox{d}\boldsymbol{\mathtt{\omega}} \quad \mbox{ for } B \in \mathcal{B}_{\mathbb{T}^2}.
\end{equation}
Note that $\tilde{\lambda}_{\rho}$ acts as a local smoothing measure at $\boldsymbol{\mathtt{\omega}}_0$ which is absolutely continuous with respect to $\lambda$, the Lebesgue measure on $\mathbb{T}^2$. 
Thus with respect to $\tilde{\lambda}_\rho$ we can define a model $p(\boldsymbol{\mathtt{\omega}}|\boldsymbol{\beta},\eta)$ on $\mathbb{T}^2$ by,
\begin{equation}\label{main_model}
p(\boldsymbol{\mathtt{\omega}}|\boldsymbol{\beta},\eta) \propto \exp\left(\boldsymbol{\beta}'T(\boldsymbol{\mathtt{\omega}}) + \eta|\boldsymbol{\beta}'T(\boldsymbol{\mathtt{\omega}})|\right)
\end{equation}
where $T$ is a function defined with range $(-\pi,\pi)$, given by $T(\mbt{\omega}) \triangleq (\tau(z), \tau(w))$ where $\tau(z) = x - 2\pi I( x > \pi),$ if $z= \exp( 2\pi j x)$ for $x \in [0,1]$. Note that the definition is ambiguous at $z = -1$ or $w=-1$, that is, $(\lbrace-1\rbrace \times \mathbb{T}) \cup (\mathbb{T}\times \lbrace-1\rbrace)$, which is a set of measure zero. Furthermore we have, 
\begin{equation}
\int_{\mathbb{T}^2} p(\boldsymbol{\mathtt{\omega}}|\boldsymbol{\beta},\eta) \mbox{d}\boldsymbol{\mathtt{\omega}} = 1 \;\;\text{for all} \;\;\boldsymbol{\beta} \in \mathbb{R}^2, \eta \in \mathbb{R}
\end{equation}
We use this  to construct local template supported on $\tilde{N}(\mbt{\omega}_0)$. These models \eqref{main_model} are \textit{scalable} and in a local neighborhood $\tilde{N}(\mbt{\omega}_0)$ it can be scaled using the \textit{support function} $\rho$ in the following manner.
\begin{equation}\label{ltm}
p(\boldsymbol{\mathtt{\omega}}|\boldsymbol{\beta},\eta) = \exp\lbrace \boldsymbol{\beta}'T(\boldsymbol{\mathtt{\omega}}\boldsymbol{\mathtt{\omega}}_0^{-1}) + \eta\left|\boldsymbol{\beta}'T(\boldsymbol{\mathtt{\omega}}\boldsymbol{\mathtt{\omega}}_0^{-1})\right| - d(\boldsymbol{\beta},\eta)\rbrace \times\rho(\boldsymbol{\mathtt{\omega}})
\end{equation}

We call this model, the \textit{Local Template Model} on $\tilde{N}(\boldsymbol{\mathtt{\omega}}_0)$. Note that we have used this model to test for edges, because it has the capability to capture the several different kinds of discontinuities. Observe that if $\eta = 0$ we have a continuous function, while the presence of $\eta$ shows the evidence of the discontinuity. The different kinds of discontinuities captured by this model, is shown in Figure 2 for varying values of $(\boldsymbol{\beta},\eta)$. Note that all directional edges can be identified through our model. It is also possible to identify other types of edges such as the `Y' shaped edge by increasing the parameters in the model.

\begin{figure}[!h]
\centering
\includegraphics[scale=0.5]{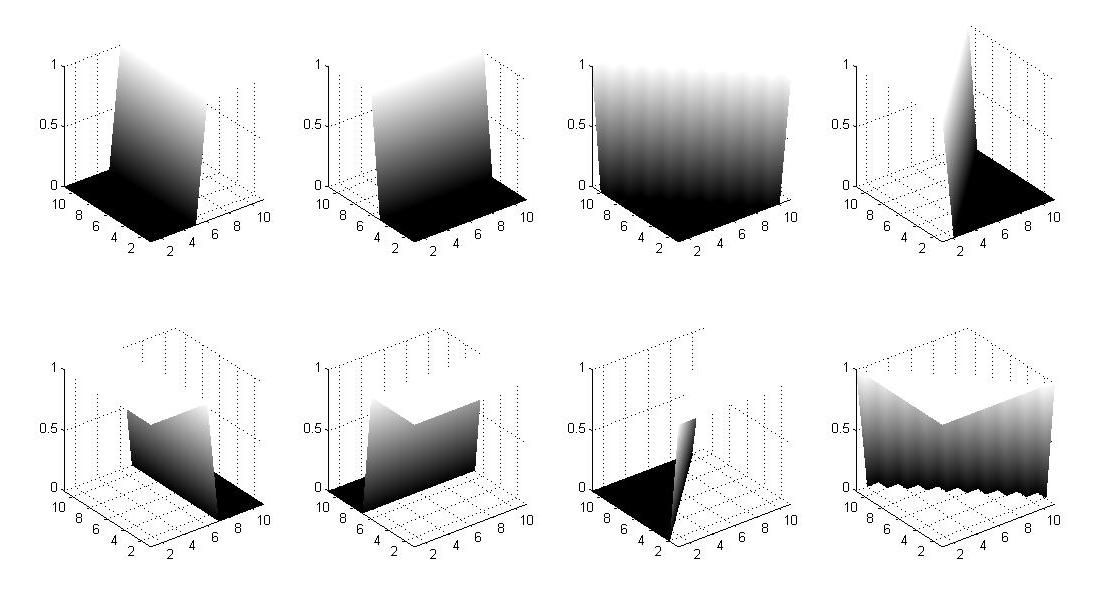}
\caption{Different Kinds of Edges detected by the Local Template Model}
\end{figure}

We first use this model to test for edges at $\boldsymbol{\mathtt{\omega}}_0$. In order to test for edges, we test $H_0 : \eta = 0$ vs. $H_1 : \eta \neq 0$. We perform usual likelihood ratio test \cite{TSH} and calculate the $p$-value. We shift the LTM over the whole image and keep track of all the $p$-values. Using these collection of $p$-values, we perform multiple testing at level $\alpha$ by using the Holms procedure \cite{Holms} to finally obtain the neighborhoods containing the edges. Thus, we obtain the regions, say $\tilde{N}_1, \tilde{N}_2, \ldots, \tilde{N}_k$ which contain edges at each of its center. Let us denote the support function in each of these neighborhoods by $\rho_i$ for $i \in \lbrace1,\ldots k\rbrace$. Note that each $\rho_i$ takes function value 0 outside of $\tilde{N}_i$ for $i \in \left\lbrace 1,\ldots, k\right\rbrace$.

We try to fit the LTM on these neighborhoods. The LTM $p$ acts as a discrete multinomial distribution on $\tilde{N}(\boldsymbol{\mathtt{\omega}}_0)$, when restricted to the pixel sites $\boldsymbol{\mathtt{\omega}}_1, \boldsymbol{\mathtt{\omega}}_2, \ldots, \boldsymbol{\mathtt{\omega}}_D$. Thus, denoting the particular $\rho_i$ by $\rho$ for ease of notation, we have,
\begin{equation}
d(\boldsymbol{\beta},\eta) = \log \left(\sum_{\boldsymbol{\mathtt{\omega}}_i \in \tilde{N}(\boldsymbol{\mathtt{\omega}}_0)} \rho(\boldsymbol{\mathtt{\omega}}_i) \times \exp \left\lbrace \boldsymbol{\beta}'T(\boldsymbol{\mathtt{\omega}}_i\mbt{\omega}_0^{-1}) + \eta|\boldsymbol{\beta}'T(\boldsymbol{\mathtt{\omega}}_i\mbt{\omega}_0^{-1})| \right\rbrace\right)
\end{equation}
Denoting $S = \sum_{\boldsymbol{\mathtt{\omega}}_i \in \tilde{N}(\boldsymbol{\mathtt{\omega}}_0)}Y(\boldsymbol{\mathtt{\omega}}_i)\rho(\boldsymbol{\mathtt{\omega}}_i)$, we have the likelihood function as
\begin{equation}\label{mle}
\begin{split}
L(\boldsymbol{\beta},\eta|\boldsymbol{Y}) &\propto \prod_{\boldsymbol{\mathtt{\omega}}_i \in \tilde{N}(\boldsymbol{\mathtt{\omega}}_0)} p\left(\boldsymbol{\mathtt{\omega}}_i|\boldsymbol{\beta},\eta\right)^{Y(\boldsymbol{\mathtt{\omega}}_i)\rho(\boldsymbol{\mathtt{\omega}}_i)/S}\\
\log(L(\boldsymbol{\beta},\eta|\boldsymbol{Y})) &= C + \sum_{\boldsymbol{\mathtt{\omega}}_i \in \tilde{N}(\boldsymbol{\mathtt{\omega}}_0)}\left(\boldsymbol{\beta}'T(\boldsymbol{\mathtt{\omega}}_i\mbt{\omega}_0^{-1}) + \eta|\boldsymbol{\beta}'T(\boldsymbol{\mathtt{\omega}}_i\mbt{\omega}_0^{-1})|\right)\times \frac{Y(\boldsymbol{\mathtt{\omega}}_i)\rho(\boldsymbol{\mathtt{\omega}}_i)}{S} - d(\boldsymbol{\beta},\eta)
\end{split}
\end{equation}
where $C$ is a constant. Thus we maximize the last equation in \eqref{mle} to obtain $(\hat{\boldsymbol{\beta}},\hat{\eta})$. Denoting this as $(\hat{\boldsymbol{\beta_i}},\hat{\eta}_i)$ for the $i$-th neighborhood, our local estimate is $p(\cdot|\hat{\boldsymbol{\beta_i}},\hat{\eta}_i)$. Transforming this back to the spatial domain, we obtain the ML estimated fit $f_i(\cdot|\hat{\theta}_i)$, of \eqref{model} for $i = 1,\ldots,k$.

\subsection{Optimal Partitions of Unity}
Using the LTM as explained in the previous subsection, we obtain $\tilde{N}_1, \tilde{N}_2, \ldots, \tilde{N}_k$ as the significant edge covering neighborhoods. The local functions are denoted by $\rho_i$ for $i \in \lbrace1,\ldots k\rbrace$. Using these local functions we create a continuous function on the whole image, such that it takes values on $\tilde{N}_i$ for $i \in \lbrace1,\ldots k\rbrace$ and 0 outside. In order to create such a function we scale each $\rho_i$ by a constant $\alpha_i$ and consider their linear sum. The value of $\alpha_i$ for $i \in \lbrace1,\ldots k\rbrace$ is obtained by solving an optimization problem, which can be stated as follows.

Find $\alpha_i$ for $i \in \lbrace1,\ldots k\rbrace$ in order to
\[\begin{array}{rl}
    \text{Maximize} &t\\
    \mbox{subject to}
    & \sum_{i=1}^{k} \alpha_i \rho_i(\boldsymbol{\mathtt{\omega}}_{0j}) \geq t\;\; \forall j \in {1,\ldots,k}\\
                & \sum_{i=1}^{k} \alpha_i \rho_i(\boldsymbol{\mathtt{\omega}}_{0j}) \leq 1\;\; \forall j \in {1,\ldots,k}\\
                & \textnormal{and} \;\; \alpha_i \geq 0 \;\;\forall i \in {1,\ldots,k}
\end{array}
\]
where $\boldsymbol{\mathtt{\omega}}_{0j}$ denotes the center of the $\tilde{N}_j$ for $j = 1, \ldots, k$. 

Denoting the solution to this optimization process as $\hat{\alpha_i}$ for $i \in \lbrace1,\ldots k\rbrace$ we get a function covering the edge points as $P_{e} = \sum_{i = 1}^{l}\hat\alpha_i\rho_i$. This is a continuous function defined on the whole space which covers all the detected edges and takes the maximum value of 1. $P_s = 1 - P_e$, correspondingly covers the smooth region of the image. Furthermore, $\{P_{e}, P_{s}\}$ forms the optimal partition of unity \cite{Munkres}. Thus the final smooth edge-preserving estimate can be written as
\begin{equation}
\hat{f} = \widehat{P_{e}f} + \widehat{P_sf}
\end{equation}
where $\widehat{P_{e}f}$ covers the local features and $\widehat{P_sf}$ denotes the smooth estimate. 

As explained in the previous subsection, if we denote the local estimate of $f\rho_i$ by $\hat{f}_i(\cdot|\hat{\theta}_i)$ as in \eqref{model}, then $\widehat{P_{e}f}$ is given by
\begin{equation}
\widehat{P_{e}f} = \sum_{i=1}^{k}\hat\alpha_i \widehat{f\rho_i} = \sum_{i=1}^{k}\hat\alpha_i \hat{f}_i(\cdot|\hat{\theta}_i)
\end{equation}
which form the first part of the estimate as in \eqref{model}. Now with this estimate in hand, we proceed to estimate the smooth region using the method of TPS regularization via Fourier basis. 

\subsection{Thin Plate Spline (TPS) Regularization via Fourier basis}
We have already discussed the details of the Fourier theory as well as the embedding of the image in the spectral domain in section 2.4. Following the notation therein, we proceed to solve the problem of smooth density estimation by minimizing \eqref{main_opt}.

Following the univariate definition in \eqref{fourier_def}, we can write the Fourier coefficients $u_{k,l}$ of a function $f \in L_2(\mathbb{T}^2,\mbox{d}\boldsymbol{\mathtt{\omega}})$ as,
\begin{equation}
u_{k,l} = \frac{1}{4\pi^2}\int_{\mathbb{T}^2}z^{-k}w^{-l}f(\boldsymbol{\mathtt{\omega}})\;\mbox{d}\boldsymbol{\mathtt{\omega}}\qquad\text{for}\;\; k,l \in \mathbb{Z}
\end{equation}
Now, square summability of the solution of \eqref{main_opt} implies the existence of a density $f$ by the Fourier Inversion Theorem. Thus, by inverse Fourier transformation and considering the density to be real, we can say
\[\begin{split}
f(\boldsymbol{\mathtt{\omega}}) &= \sum_{k=-\infty}^{\infty}\sum_{l=-\infty}^{\infty}u_{k,l} z^{k}w^{l}\\
&=\sum_{k=-\infty}^{\infty}\left( \sum_{l=-\infty}^{-1}u_{k,l}z^{k}w^{l} + u_{k,0}z^{k} +  \sum_{l=1}^{\infty}u_{k,l}z^{k}w^{l}\right)
\end{split}
\]
Now using $u_{-k,-l} = \overline{u_{k,l}}; u_{-k,l} = \overline{u_{k,-l}}$, the above is simplified to,
\begin{equation}\label{fourier_density}
f(\boldsymbol{\mathtt{\omega}}) = u_{0,0} + 2\times\Re\left(\sum_{k=1}^{\infty}u_{k,0}z^{k}  + \sum_{l=1}^{\infty}u_{0,l}w^{l} + \sum_{k=1}^{\infty}\sum_{l=1}^{\infty}u_{k,l} z^{k}w^{l} + \sum_{k=1}^{\infty}\sum_{l=1}^{\infty}u_{k,-l} z^{k}w^{-l}\right)
\end{equation}
where $\Re(z)$ denotes the real part of the complex number $z$.

Thus, the penalty term becomes, $\lambda \int_{\mathbb{T}^2} \left(f_{zz}^2 + f_{ww}^2 + 2f_{zw}^2\right) \mbox{d}\boldsymbol{\omega}$ and the problem of density estimation now becomes the minimization problem of
\begin{equation}
\frac{1}{4\pi^2}\int_{\mathbb{T}^2}|\rho_0(\boldsymbol{\mathtt{\omega}})\,\femp (\boldsymbol{\mathtt{\omega}}) - f(\boldsymbol{\mathtt{\omega}})|^2 \mbox{d}\boldsymbol{\omega} + \lambda \int_{\mathbb{T}^2} \left(f_{zz}^2 + f_{ww}^2 + 2f_{zw}^2\right) \mbox{d}\boldsymbol{\omega}
\end{equation}
in the class $L^2(\mathbb{T}^2, \mbox{d}\boldsymbol{\omega})$. Here we retain the same notation of the empirical even after embedding into $\mathbb{T}^2$. Now putting $u_{k,l} = x_{k,l} + jy_{k,l}$, the problem in the thin plate spline format becomes the minimization problem of
\begin{equation}\label{final_opt}
\begin{split}
\sum_{k=-\infty}^{\infty}\sum_{l=-\infty}^{\infty}|\hat{u}_{k,l} - u_{k,l}|^2  + \lambda\left\lbrace \sum_{k=1}^{\infty} k^4(x_{k,0}^2 + y_{k,0}^2) + \sum_{k=1}^{\infty} l^4(x_{0,l}^2 + y_{0,l}^2)\right\rbrace \\
+ \lambda\left\lbrace\sum_{k=1}^{\infty}\sum_{l=1}^{\infty}(k^2  + l^2)^2\times(x_{k,l}^2 + y_{k,l}^2 + x_{k,-l}^2 + y_{k,-l}^2)\right\rbrace
\end{split}
\end{equation}
Thus, we get a 1-1 correspondence between the curve fitting problem via thin plate spline regularization and the density estimation problem using the Fourier basis. The solution to this minimization problem is given as a theorem below.

\begin{theorem}\label{soln}
Solution to the minimization problem given by equation \eqref{final_opt} above gives rise to a kernel like density estimate given by
\[\hat{f}(\boldsymbol{\mathtt{\omega}}) = \hat{u}_{0,0} + 2\Re\left(\sum_{k=1}^{\infty}\frac{\hat{u}_{k,0}z^{k} }{1 + \lambda k^4}  + \sum_{l=1}^{\infty}\frac{\hat{u}_{0,l} w^{l}}{1 + \lambda l^4} + \sum_{k=1}^{\infty}\sum_{l=1}^{\infty}\frac{\hat{u}_{k,l} z^{k}w^{l} + \hat{u}_{k,-l} z^{k}w^{-l}}{1 + \lambda(k^2 + l^2)^2}\right)\]
\end{theorem}
\begin{proof}
See Appendix.
\end{proof}

The optimal value of $\lambda$ is obtained via a grid search, which minimizes, the integrated mean squared error. That is, we choose that value of $\lambda$ for which,
\begin{equation}
\int_{\mathbb{T}^2}\mathbb{E}\big|\hat{f}_\lambda(\boldsymbol{\mathtt{\omega}}) - f(\boldsymbol{\mathtt{\omega}})\big|^2\mbox{d}\boldsymbol{\omega}
\end{equation}
is minimized. The estimate of the smooth region in $\mathcal{D}$ is obtained by the inverse transformation of the density estimate of Theorem \eqref{soln}. 

\section{Implementation}
In this section we apply our methodology on the image of Lenna of size $512 \times 512$. We have chosen the Local Template Model on an $11\times 11$ square. We first define a function $\tilde{\rho}$ on the window and transform it into our required function $\rho$ on $\mathcal{D}$ and subsequently on $\mathbb{T}^2$. The choice of the local function $\tilde{\rho}$ on the $11\times 11$ window is delicate. We try to put maximum weight on the center and gradually decrease it. The function $\tilde{\rho}$ is defined as the convolution of two functions $h_1$ and $h_2$ each defined on a $6\times 6$ window. That is 
\begin{equation}
\tilde{\rho}(n_1,n_2) = \sum_{k_1 = -\infty}^{\infty}\sum_{k_2 = -\infty}^{\infty}h_1(k_1,k_2)h_2(n_1 - k_1, n_2 - k_2)
\end{equation}
for $n_1,n_2 \in \left\lbrace 1,\ldots,11\right\rbrace $. The functions $h_1$ and $h_2$ are defined as follows. 
\begin{equation}
h_1(n_1,n_2) = 
\begin{cases}
0 & n_1,n_2 \not\in \left\lbrace 2,\ldots,5\right\rbrace\\
0.5 & n_1 = \left\lbrace 2,5\right\rbrace\;\; n_2 = \left\lbrace 2,\ldots,5\right\rbrace\\
0.5 & n_2 = \left\lbrace 2,5\right\rbrace\;\; n_1 = \left\lbrace 2,\ldots,5\right\rbrace\\
1 & \text{otherwise}
\end{cases}
\end{equation}
which is a trapezoid on the $6\times 6$ window. Denoting the center $(3.5,3.5)$ as $(c_1,c_2)$, we define,
\begin{equation}
h_2(n_1,n_2) = 
\begin{cases}
\exp\left\lbrace -\tan\left(\frac{\pi(n_1-c_1)}{2\tau}\right)^2 \right\rbrace \times \exp\left\lbrace -\tan\left(\frac{\pi(n_2-c_2)}{2\tau}\right)^2 \right\rbrace  & n_1,n_2 \in \left\lbrace 1,\ldots,6\right\rbrace\\
0 & \text{otherwise}
\end{cases}
\end{equation}
We have chosen the value of $\tau = 5$. However, this choice can be changed and is left upto the user. $\tau$ controls the spread of the final convoluted function $\tilde{\rho}$. Figure 3, shows the graph of $\tilde{\rho}$. This function is rescaled into $\mathcal{D}$ and subsequently to $\mathbb{T}^2$ to get the local function $\rho$.

\begin{figure}[!h]
\centering
\includegraphics[scale=0.55]{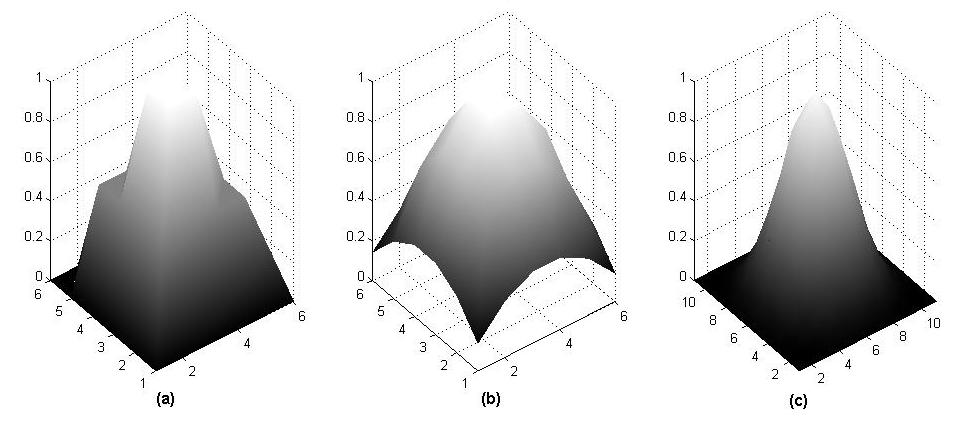}
\caption{(a) The trapezoid function on square of size $6\times 6$; (b) The function $h(x,y)$ on a square of size $6\times 6$; (c) The final convoluted function $\tilde{\rho}(x,y)$ for the $11\times 11$ window}
\end{figure}

Using this local function $\rho$ in the LTM, we execute the edge detection algorithm as explained in section 3.1. We start with the center at pixel position $(6,6)$ and shift the template by 3 pixels for the next iteration. We keep track of all the $p$-values from the LRT, and finally obtain the edge pixels using the Holms procedure of multiple testing at level $\alpha = 0.01$. Note that the extreme borders of the image are inherently considered as edges. Using the detected and inherent windows we create the partition of unity by solving the linear programming problem stated in section 3.2. Let us denote the function covering the detected edge pixels by $P_{e}$. Now, instead of showing just the centre of the detected window as edge points, we perform a neat trick to get the edge lines.

For each detected edge pixel $\boldsymbol{\mathtt{\omega}}_0$, we mark the pixels in $N(\boldsymbol{\mathtt{\omega}}_0)$ which lie on the straight line, $\hat{\boldsymbol{\beta}}'(\boldsymbol{\mathtt{\omega}} - \boldsymbol{\mathtt{\omega}}_0)$, where $\hat{\boldsymbol{\beta}}$ are obtained by maximizing the likelihood in the LTM model restricted to $N(\boldsymbol{\mathtt{\omega}}_0)$. This gives the direction of the edge in $N(\boldsymbol{\mathtt{\omega}}_0)$. Now we put additional weights on these pixel positions by the function $P_{e}$. It has been seen experimentally, that all points with function value greater than 0.8 gives the best visual representation of the edge locations. Figure 4 shows the original image of Lenna, followed by the detected edges in Figure 5.

\begin{figure}[!h]
\begin{minipage}[b]{0.5\linewidth}
\centering
\includegraphics[width = 5cm]{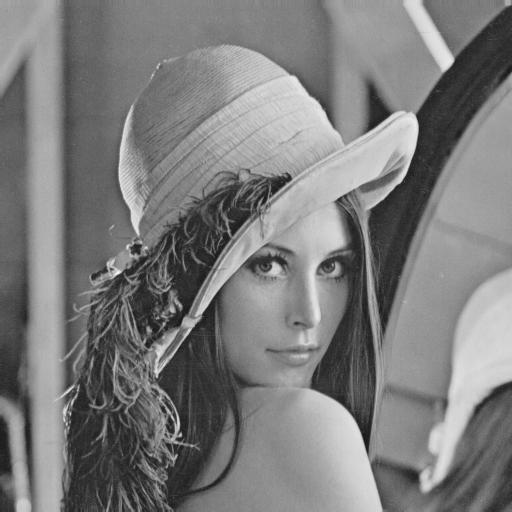}
\caption{Original Image of Lenna}
\end{minipage}
\hspace{0.5cm}
\begin{minipage}[b]{0.5\linewidth}
\centering
\includegraphics[width = 5cm]{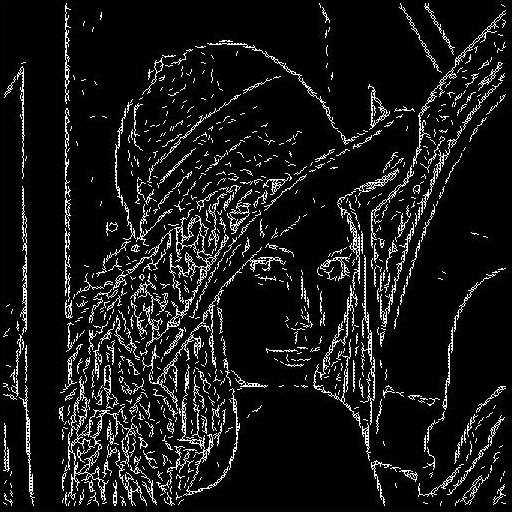}
\caption{Image after edge detection using LTM}
\end{minipage}
\end{figure}

We fit the LTM on the detected edges, to get the edge estimates using the method explained in Section 3.1. The fitted edges are shown in Figure 6. The remaining image, i.e. $P_sf$ (shown in figure 7), is estimated using the TPS regularization via the Fourier basis. The black patches in Figure 7 denotes the regions of the edges which have been omitted in the Fourier application to prevent Gibbs phenomenon \cite{Hewitt} from occurring.

\begin{figure}[!h]
\begin{minipage}[b]{0.5\linewidth}
\centering
\includegraphics[width = 5cm]{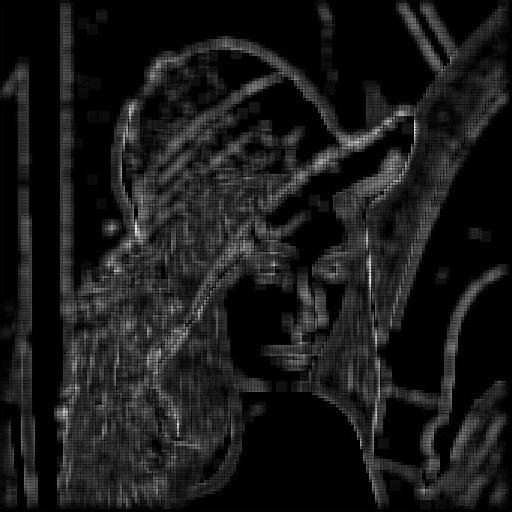}
\caption{Fitted Edges using the Local Template Model}
\end{minipage}
\hspace{0.5cm}
\begin{minipage}[b]{0.5\linewidth}
\centering
\includegraphics[width = 5cm]{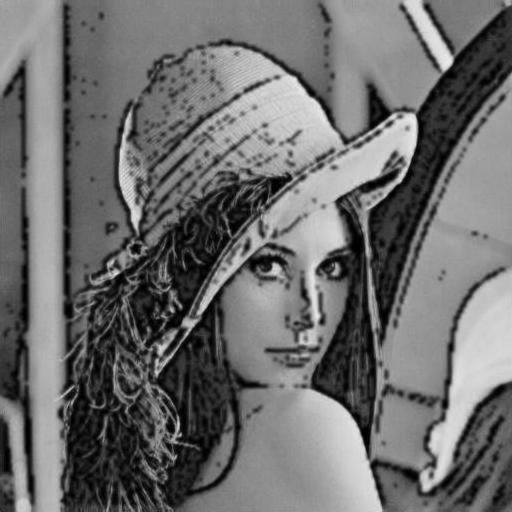}
\caption{Input image for the Fourier technique}
\end{minipage}
\end{figure}

Figure 8 shows the output from the Fourier technique as explained in Section 3.3. The final smooth edge-preserving density estimate is shown in Figure 9. This is obtained by adding the two estimates, viz. $\widehat{P_{e}f} + \widehat{P_sf}$.

\begin{figure}[!h]
\begin{minipage}[b]{0.5\linewidth}
\centering
\includegraphics[width = 5cm]{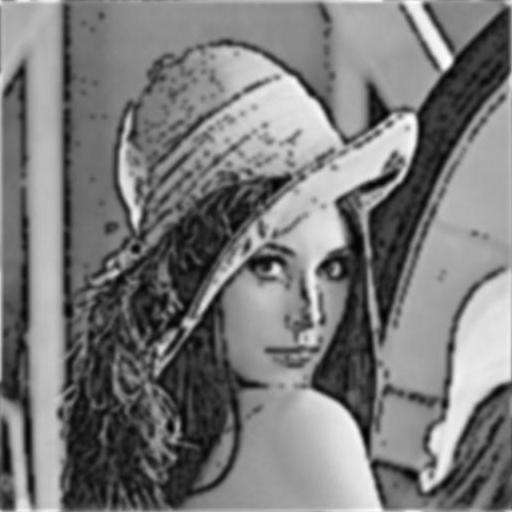}
\caption{The estimate through Penalized bivariate Fourier Transform}
\end{minipage}
\hspace{0.5cm}
\begin{minipage}[b]{0.5\linewidth}
\centering
\includegraphics[width = 5cm]{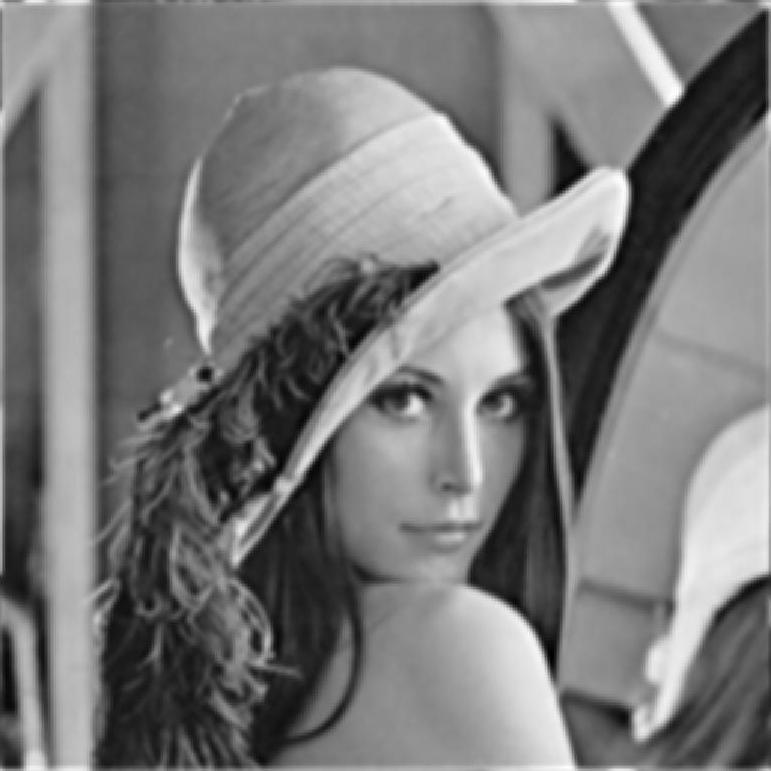}
\caption{Final Smooth Edge-Preserving Density Estimate}
\end{minipage}
\end{figure}

\section{Experimental Results}
\subsection{Outputs from SSH Algorithm on Noisy Images}
A quick browse through the literature in this field shows that there have been several methods which have been proposed for dealing with edge-preserving function estimation for noisy images. Few such methods are based on M Estimators \cite{Chu}, TM Smoothers \cite{Hillebrand}, 2-D discrete wavelet transforms \cite{Chaganti}, etc. The fact against using Fourier transform was that while it smoothens the noise, the edges are not well preserved because of the Gibbs phenomenon. However, in our methodology, since we separate the edges before going into the TPS regularization based on Fourier basis, it is only logical to see how well, our method works for noisy images.

Instead of adding artificial salt and pepper noise to the data, we create noisy images through repeated simulations by Gibbs sampling. We consider the original image to be a bi-variate distribution and using Gibbs sampling we draw a large enough sample of size $n = 50\times image \;size$, from this distribution. The frequency plot of these points, give us an image. We consider this as the noisy image and proceed with our methodology. The initial original and noisy images are shown in Figure 10.

\begin{figure}[!h]
\centering
\includegraphics[scale=0.7]{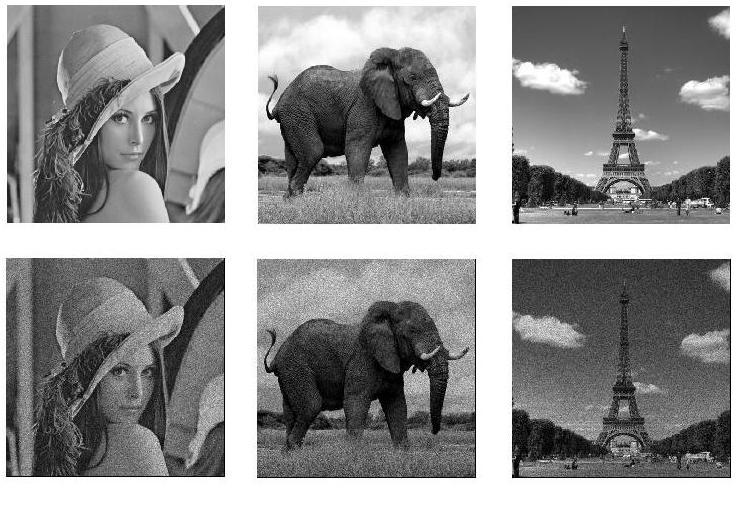}
\caption{The first row shows the original images. The second row shows the images with noise from the Gibbs sampling}
\end{figure}

Using these noisy images we execute our algorithm. The final results are shown in Figures 11,12 and 13. Note that the edge detection methodology captures most of the noise and removes it from the image during edge estimation. As a result when applying the Fourier transformations, the noisy regions are skipped. This results in a smooth and less noisy edge-preserving image in the final reconstruction, as it can be seen from the final estimates of each of the three figures. Hence, our image reconstruction methodology is somewhat robust to artificial noise, which is not the case for the direct use of the Fourier transformation.

\begin{figure}[!h]
\centering
\includegraphics[scale=0.4]{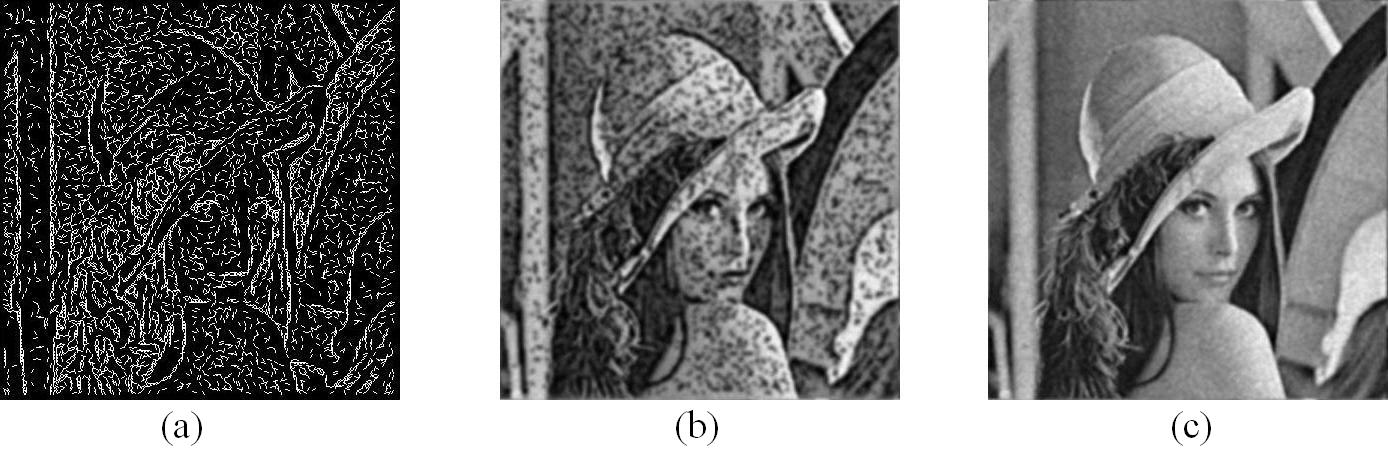}
\caption{(a) Edges Estimated in the Noisy Image of Lenna; (b) The Smooth Estimate through Fourier transform; (c) The final edge-preserving smooth density Estimate}
\end{figure}

\begin{figure}[!h]
\centering
\includegraphics[scale=0.4]{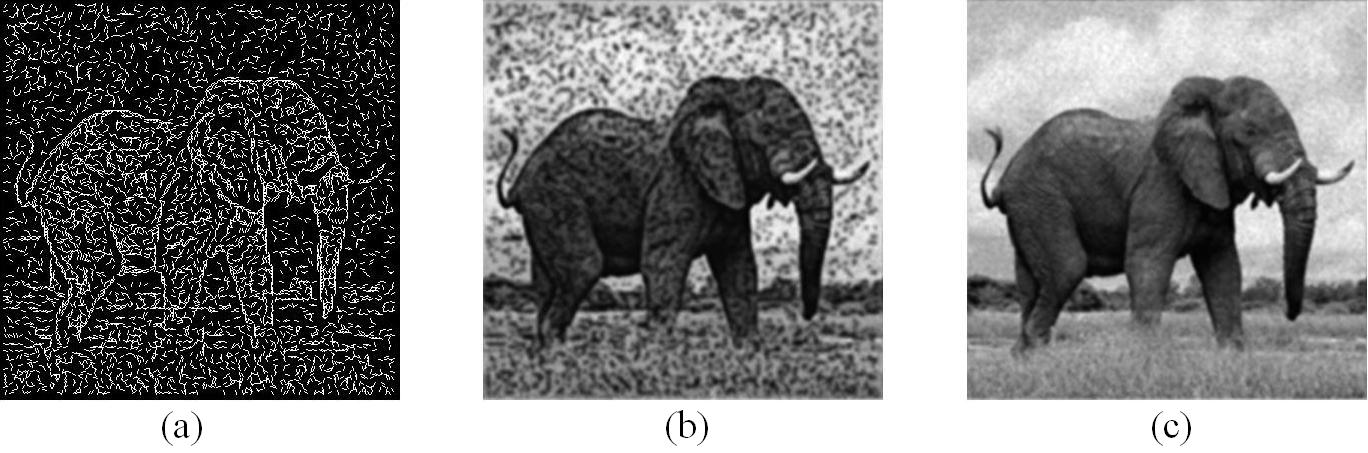}
\caption{(a) Edges Estimated in the Noisy Image of African Elephant; (b) The Smooth Estimate through Fourier transform; (c) The final edge-preserving smooth density Estimate}
\end{figure}

\begin{figure}[!h]
\centering
\includegraphics[scale=0.4]{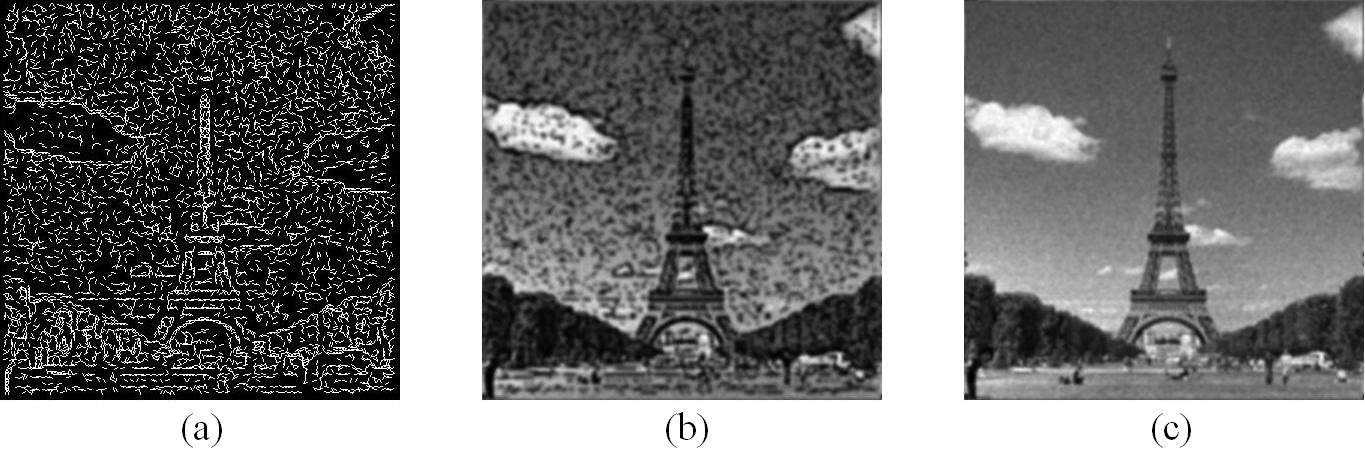}
\caption{(a) Edges Estimated in the Noisy Image of Eiffel Tower; (b) The Smooth Estimate through Fourier transform; (c) The final edge-preserving smooth density Estimate}
\end{figure}

\subsection{Error Analysis}
As explained in the previous subsection, using the Gibbs Sampling technique we draw a sample of size $T$ from the original bi-variate distribution. These $T$ points give us an image. We repeat this procedure to obtain $N = 100$ such images. For each of these images we estimate the final edge-preserving smooth density. The choice of the sample size in our simulations is an important factor. Considering the image to be of size $k\times k$ we choose the sample sizes as $T = m \times k \times k$ where $m = 10, 20, 50$ and $100$.

In this subsection we calculate the discretized mean square error (DMSE), introduced in \eqref{dmse_def}, of this estimate and tabulate it in Table 1. We also calculate and report the \textit{within sample variance} of the estimate in Table 2. Furthermore, the ratio of the DMSE to the within sample variance gives us an idea about the coefficient of variation which we report in Table 3.

\begin{table}[!h]
\centering
\caption{Discretized Mean Square Error (DMSE)}
\begin{tabular}{|@{\hspace{0.25cm}}c@{\hspace{0.25cm}}|@{\hspace{0.25cm}}c@{\hspace{0.25cm}}|@{\hspace{0.25cm}}c@{\hspace{0.25cm}}|@{\hspace{0.25cm}}c@{\hspace{0.25cm}}|}
\hline
Sample Size & Lenna's & African & Eiffel\\
$m$ $(T = mk^2)$  &  Image & Elephant & Tower\\
\hline
10&0.1773&0.1643&0.1803\\
20&0.1164&0.1034&0.1277\\
50&0.0834&0.0758&0.0878\\
100&0.0712&0.0685&0.0731\\
\hline
\end{tabular}
\end{table}

\begin{table}[!h]
\centering
\caption{Within sample variance for the different images}
\begin{tabular}{|@{\hspace{0.25cm}}c@{\hspace{0.25cm}}|@{\hspace{0.25cm}}c@{\hspace{0.25cm}}|@{\hspace{0.25cm}}c@{\hspace{0.25cm}}|@{\hspace{0.25cm}}c@{\hspace{0.25cm}}|}
\hline
Sample Size & Lenna's & African & Eiffel\\
$m$ $(T = mk^2)$  &  Image & Elephant & Tower\\
\hline
10&0.0936&0.0901&0.1006\\
20&0.0452&0.0395&0.0531\\
50&0.0179&0.0121&0.0231\\
100&0.0090&0.0089&0.0102\\
\hline
\end{tabular}
\end{table}

\begin{table}[!h]
\centering
\caption{Ratio of DMSE to Within Sample Variance}
\begin{tabular}{|@{\hspace{0.25cm}}c@{\hspace{0.25cm}}|@{\hspace{0.25cm}}c@{\hspace{0.25cm}}|@{\hspace{0.25cm}}c@{\hspace{0.25cm}}|@{\hspace{0.25cm}}c@{\hspace{0.25cm}}|}
\hline
Sample Size & Lenna's & African & Eiffel\\
$m$ $(T = mk^2)$  &  Image & Elephant & Tower\\
\hline
10&1.8950&1.8235&1.7922\\
20&2.5777&2.6177&2.4048\\
50&4.6667&6.2644&3.8008\\
100&7.9427&7.6966&7.1667\\
\hline
\end{tabular}
\end{table}

We also report the average time taken by our SSH Algorithm. All our simulations have been performed in MATLAB 7.10.0 on Windows platform, with a 4GB RAM machine and Intel Core 2 Duo 2.5 GHz Processor. The average time for detecting the edges in a $512 \times 512$ image using a window of size $11\times 11$ took about 45 seconds, followed by the local template fitting accounting for another 30 seconds. During this time, we simultaneously create the optimal partition of unity. The Fourier methodology via TPS regularization took about 2 mins to complete. Thereby we get the whole density estimate in about 3 minute and 10 seconds on the average.

\section{Some Concluding Remarks}
In this paper, we have discussed an unsupervised novel statistical methodology for image restoration and compression by using local parametric mixtures and Lagrangian relaxation. Considering the thin plate spline regularization we showed the equivalence of the density estimation problem to the TPS problem. Since Fourier compression fails to restore a piece-wise smooth image due to Gibbs phenomenon, we adapt our methodology by using the topological concept of \textit{partition of unity}.
	
We consider the image as a histogram data from a bi-variate distribution. Edge detection and estimation is performed through a local parametric model on a small window, which we term as the Local Template Model. Holms procedure of multiple testing is used to obtain the regions containing the edge pixels. The optimal partition of unity is created by proper scaling of the local weight functions. After edge estimation through the LTM, the remaining smooth region is estimated by the penalized bi-variate Fourier transform. The final estimate is obtained by adding the local and the smooth estimates.

We implement our algorithm on the image of Lenna, and we are able to get the output in two distinct channels, viz. the edge pixels and the Fourier coefficients with the thresholding sequence. We have also performed an error analysis of our methodology. The final reconstructed image is smooth and edge-preserving, even in the presence of 'salt and pepper' noise. Using repeated simulations, we have also calculated the coefficient of variation for each of the three layers of the output - the edge pixels, the Fourier smoothing and the final restored image.

One of the practical usefulness of Fourier transforms depend on their fast convergence to the limit in \eqref{fourier_inv} which depends on the rate of decay of Fourier coefficients. As can be seen form \eqref{fourier_deriv}, the roughness in the signal reflected by high $L_2$ norm of the derivatives creates a hindrance for quick approximation. The approximations also tend to be under-smoothed in case a large number of terms need to be added in the right hand side of \eqref{fourier_inv} for an accurate approximation for noisy signals. Fast decaying filters applied to the raw coefficients increases the bias. The worst case scenario occurs when the original signal is piece-wise smooth with jump discontinuities. In such cases the approximation suffers from what is known as Gibbs phenomenon. Details of Gibbs phenomenon for the standard saw-tooth signals can be found in \cite{Hewitt,Korner}. One of the objections raised against spectral embedding of spatial signals is that the decay of the coefficients slows down due to unknown location of discontinuities. One known discontinuity always exists at $z=1$ unless the signal is periodic. This is called \textit{edge effect}. It is not a difficult task to remove the local ripple due to the edge discontinuity. One of the novelty of the method presented in this article is to demonstrate that we can extract the smooth periodic part of the signal after eliminating unknown number of local ripples after searching through the significant the signal using the parametric local templates quite effectively. The Fourier approximation performs quite well on the extracted smooth  periodic portion of the signal even under noisy environments.

Although our methodology works well, it has few a limitation. The parametric model that we have chosen for the LTM, has the ability to capture any linear directional curved edges, as seen in Figure 1. However, in to capture a `Y' shaped edge we need to add more parameters to the model. Furthermore, our edge estimation technique fails to capture neat edges in the presence of noise. But, that is not a hindrance, as the final reconstructed image even in the presence of noise is a smooth and edge preserving estimate, which was our initial goal.

We hope this work helps future researchers in the endeavors especially in the field of image restoration and compression.

\section*{Appendix}
In this appendix we provide the detailed proof of Theorem \ref{soln} and the complete algorithm for the edge-preserving smooth density estimate.
\subsection*{Proof of Theorem \ref{soln}}
\begin{proof}
We need to minimize,
\begin{equation}
\begin{split}
\sum_{k=-\infty}^{\infty}\sum_{l=-\infty}^{\infty}|\hat{u}_{k,l} - u_{k,l}|^2  + \lambda\left\lbrace \sum_{k=1}^{\infty} k^4(x_{k,0}^2 + y_{k,0}^2) + \sum_{k=1}^{\infty} l^4(x_{0,l}^2 + y_{0,l}^2)\right\rbrace \\
+ \lambda\left\lbrace\sum_{k=1}^{\infty}\sum_{l=1}^{\infty}(k^2  + l^2)^2\times(x_{k,l}^2 + y_{k,l}^2 + x_{k,-l}^2 + y_{k,-l}^2)\right\rbrace
\end{split}
\end{equation}
Using $u_{-k,-l} = \overline{u_{k,l}}; u_{-k,l} = \overline{u_{k,-l}}$, the first term of the above expression can be simplified as follows

\[
\begin{split}
\sum_{k,l=-\infty}^{\infty}|\hat{u}_{k,l} - u_{k,l}|^2 &= \sum_{k=-\infty}^{\infty}\sum_{l=1}^{\infty}|\hat{u}_{k,-l} - u_{k,-l}|^2 + \sum_{k=-\infty}^{\infty}|\hat{u}_{k,0} - u_{k,0}|^2 + \sum_{k=-\infty}^{\infty}\sum_{l=1}^{\infty}|\hat{u}_{k,l} - u_{k,l}|^2\\
&= \sum_{l=1}^{\infty}\sum_{k=1}^{\infty}|\hat{u}_{-k,-l} - u_{-k,-l}|^2 + \sum_{l=1}^{\infty}|\hat{u}_{0,-l} - u_{0,-l}|^2 + \sum_{l=1}^{\infty}\sum_{k=1}^{\infty}|\hat{u}_{k,-l} - u_{k,-l}|^2\\
&\;\;\;\;+ \sum_{k=1}^{\infty}|\hat{u}_{-k,0} - u_{-k,0}|^2 + |\hat{u}_{0,0} - u_{0,0}|^2 + \sum_{k=1}^{\infty}|\hat{u}_{k,0} - u_{k,0}|^2\\
&\;\;\;\;+\sum_{l=1}^{\infty}\sum_{k=1}^{\infty}|\hat{u}_{-k,l} - u_{-k,l}|^2 + \sum_{l=1}^{\infty}|\hat{u}_{0,l} - u_{0,l}|^2 + \sum_{l=1}^{\infty}\sum_{k=1}^{\infty}|\hat{u}_{k,l} - u_{k,l}|^2\\
\end{split}
\]
Thus we get,
\begin{equation}
\begin{split}
\sum_{k,l=-\infty}^{\infty}|\hat{u}_{k,l} - u_{k,l}|^2 &= |\hat{u}_{0,0} - u_{0,0}|^2 + 2\sum_{l=1}^{\infty}\sum_{k=1}^{\infty}|\hat{u}_{k,l} - u_{k,l}|^2 + 2\sum_{l=1}^{\infty}\sum_{k=1}^{\infty}|\hat{u}_{k,-l} - u_{k,-l}|^2\\
&\;\;\;\;+2\sum_{l=1}^{\infty}|\hat{u}_{0,l} - u_{0,l}|^2 + 2\sum_{k=1}^{\infty}|\hat{u}_{k,0} - u_{k,0}|^2
\end{split}
\end{equation}
Let $\hat{u}_{k,l} = \overline{\cos(k\theta + l\gamma)} + j\overline{\sin(k\theta + l\gamma)}$ and $u_{k,l} = x_{k,l} + jy_{k,l}$, where $\overline{\phi(\theta,\gamma)}$ denotes the mean of $\phi$ over both $\theta$ and $\gamma$. Now the problem reduces to minimizing,

\begin{equation}
\begin{split}
&\sum_{l=1}^{\infty}\sum_{k=1}^{\infty}\left\lbrace\left(\overline{\cos(k\theta + l\gamma)} - x_{k,l}\right)^2 + \left(\overline{\sin(k\theta + l\gamma)} - y_{k,l}\right)^2 + \left(\overline{\cos(k\theta - l\gamma)} - x_{k,l}\right)^2 \right\rbrace\\
&+ \sum_{l=1}^{\infty}\sum_{k=1}^{\infty}\left\lbrace + \left(\overline{\sin(k\theta - l\gamma)} - y_{k,l}\right)^2 \right\rbrace + \sum_{k=1}^{\infty}\left\lbrace\left(\overline{\cos(k\theta)} - x_{k,0}\right)^2 + \left(\overline{\sin(k\theta)} - y_{k,0}\right)^2\right\rbrace\\
&+ \sum_{l=1}^{\infty}\left\lbrace\left(\overline{\cos(l\gamma)} - x_{0,l}\right)^2 + \left(\overline{\sin(l\gamma)} - y_{0,l}\right)^2\right\rbrace + \lambda\left\lbrace \sum_{k=1}^{\infty} k^4(x_{k,0}^2 + y_{k,0}^2) + \sum_{k=1}^{\infty} l^4(x_{0,l}^2 + y_{0,l}^2)\right\rbrace \\
&+ \lambda\left\lbrace\sum_{k=1}^{\infty}\sum_{l=1}^{\infty}(k^2 +l^2)^2\times(x_{k,l}^2 + y_{k,l}^2 + x_{k,-l}^2 + y_{k,-l}^2)\right\rbrace
\end{split}
\end{equation}
Differentiating this with respect to $x_{k,l}$ and $y_{k,l}$ and for each $k,l$, and equating to zero, we get,
\begin{equation}
\begin{split}
\widehat{x_{k,0}} = \frac{\overline{\cos(k\theta)}}{1 + \lambda k^4} \qquad &\textnormal{and} \qquad \widehat{y_{k,0}} = \frac{\overline{\sin(k\theta)}}{1 + \lambda k^4}\\
\widehat{x_{0,l}} = \frac{\overline{\cos(l\gamma)}}{1 + \lambda l^4} \qquad &\textnormal{and} \qquad \widehat{y_{0,l}} = \frac{\overline{\sin(l\gamma)}}{1 + \lambda l^4}\\
\widehat{x_{k,l}} = \frac{\overline{\cos(k\theta + l\gamma)}}{1 + \lambda(k^2 +l^2)^2} \qquad &\textnormal{and} \qquad \widehat{y_{k,l}} = \frac{\overline{\sin(k\theta + l\gamma)}}{1 + \lambda(k^2 +l^2)^2}\\
\widehat{x_{k,-l}} = \frac{\overline{\cos(k\theta - l\gamma)}}{1 + \lambda(k^2 +l^2)^2} \qquad &\textnormal{and} \qquad \widehat{y_{k,-l}} = \frac{\overline{\sin(k\theta - l\gamma)}}{1 + \lambda(k^2 +l^2)^2}\\
\end{split}
\end{equation}
Thus the final estimate using \eqref{fourier_density} is
\[\hat{f}(\mbt{\omega}) = \hat{u}_{0,0} + 2\Re\left(\sum_{k=1}^{\infty}\frac{\hat{u}_{k,0}z^{k} }{1 + \lambda k^4}  + \sum_{l=1}^{\infty}\frac{\hat{u}_{0,l} w^{l}}{1 + \lambda l^4} + \sum_{k=1}^{\infty}\sum_{l=1}^{\infty}\frac{\hat{u}_{k,l} z^{k}w^{l} + \hat{u}_{k,-l} z^{k}w^{-l}}{1 + \lambda(k^2 +l^2)^2}\right)\]
\\\\
Hence, proved.
\begin{flushright}$\blacksquare$\end{flushright}

\end{proof}
\subsection*{Implementation Algorithm}
We now give a detailed algorithm of our implementation to obtain the edge preserving smooth density estimate.
\\\\
\begin{algorithm}[H]
 \SetAlgoLined
 \KwData{The grey value image $I$}
 $sz$ = Size of the image matrix\;
 $N = 5$, To create the $11\times 11$ matrix for the LTM\;
 Create the function $\tilde{\rho}$ on the $11\times 11$ window and store in $\rho$\;
 \For{i from $N+1$ to $sz_1 - N - 1$ with a jump of $N/2$}{
	\For{j from $N+1$ to $sz_2 - N - 1$ with a jump of $N/2$}{
		Put $center = (i,j)$\;
		Perform likelihood ratio test for $\eta = 0$ vs $\eta \neq 0$ using the model in \eqref{ltm} and store result in $pval(i,j)$\;
	} 
 }
Execute the Holms procedure with these $p$-values and store the result in $test$. This variable $test$ becomes a 0-1 matrix denoting all location with 1 as edges\;
Put 1 on all the extreme border pixels of the image\;
 $t_1 = $ sum total of all elements of test\;
 Solve the linear programming problem posed in section 3.2 by increasing $t$ over a grid as $0:0.005:1$ and store the result in $\alpha$ which is a vector of length $t_1$\;
 Initialize $g, P_e$ as zero matrices of size $sz$\;
 Initialize $count = 1$\; 
 
 \For{i from $N+1$ to $sz_1 - N - 1$ with a jump of $N/2$}{
	\For{j from $N+1$ to $sz_2 - N - 1$ with a jump of $N/2$}{
		Put $center = (i,j)$\;
		\If{$test(i,j) = 1$}{
			$S_{count} = $sum of pixel values in $I$ restricted to the $11\times 11$ window with center at $(i,j)$\;
			Estimate the local density using the LTM and store it in $f_{count}$\;
			Increment $g$ by $\alpha_{count}\times f_{count} \times S_{count}$\;
			Increment $P_e$ by $\alpha_{count}\times \rho$\;
			Increment $count$ by $1$\;
		}
	} 
 }
 Execute the TPS regularization based on Fourier basis as given in section 2 using the input image as $P_sI$, where $P_s = 1 - P_e$ and store the result in $h$\;
 The final estimate is $\hat{f} = g + h$.
\end{algorithm}

\end{document}